\documentclass[11pt]{article}
\usepackage{fullpage}
\usepackage[bookmarks]{hyperref}
\usepackage{amssymb}
\usepackage{amsmath}
\usepackage{amsthm}
\usepackage{youngtab}
\usepackage{graphicx,color,colordvi}
\usepackage{subfig}
\usepackage{bbm}
\usepackage{cleveref}
\usepackage[utf8]{inputenc}
\usepackage{fullpage}
\usepackage[blocks]{authblk}
\usepackage{mathtools}
\usepackage{multirow}
\usepackage{dsfont}
\usepackage[toc,page]{appendix}

\newcommand{\nc}{\newcommand}
\nc{\rnc}{\renewcommand}
\nc\mnb[1]{\medskip\noindent{\bf #1}}

\newcommand{\M}{{\mathbb{M}}}

\newcommand{\ot}{\otimes}

\newcommand{\Hom}{\operatorname{Hom}}
\renewcommand{\>}{\rangle}
\newcommand\be{\begin{equation}}
\newcommand\ee{\end{equation}}

\newtheorem{innercustomgeneric}{\customgenericname}
\providecommand{\customgenericname}{}
\newcommand{\newcustomtheorem}[2]{%
	\newenvironment{#1}[1]
	{%
		\renewcommand\customgenericname{#2}%
		\renewcommand\theinnercustomgeneric{##1}%
		\innercustomgeneric
	}
	{\endinnercustomgeneric}
}

\newcustomtheorem{customthm}{Theorem}
\newcustomtheorem{customlemma}{Lemma}
\newcustomtheorem{customprop}{Proposition}

\DeclareMathOperator{\tr}{Tr}

\usepackage{hyperref}
\hypersetup{colorlinks, citecolor=magenta, filecolor=blue, linkcolor=blue, urlcolor=green}
\newtheorem{theorem}{Theorem}

\newtheorem{corollary}[theorem]{Corollary}
\newtheorem{fact}[theorem]{Fact}

\newtheorem{definition}[theorem]{Definition}

\newtheorem*{theorem*}{Theorem}
\newtheorem{lemma}[theorem]{Lemma}
\newtheorem{notation}[theorem]{Notation}

\newtheorem{proposition}[theorem]{Proposition}
\newtheorem{remark}[theorem]{Remark}

\begin{document}
\title{Port-based teleportation in arbitrary dimension}

\author{Micha{\l} Studzi{\'n}ski}
\author{Sergii Strelchuk}
\affil[1]{\small DAMTP, Centre for Mathematical Sciences, University of Cambridge, Cambridge~CB30WA, UK}
\author{Marek Mozrzymas}
\affil[2]{\small Institute for Theoretical Physics, University of Wrocław
	50-204 Wrocław, Poland}
\author{Micha{\l} Horodecki}
\affil[3]{\small Institute of Theoretical Physics and Astrophysics, National Quantum Information Centre, Faculty of Mathematics, Physics and Informatics, University of Gda{\'n}sk, Wita Stwosza 57, 80-308 Gda{\'n}sk, Poland}
\date{}
\maketitle			 
\begin{abstract}
	Port-based teleportation (PBT), introduced in 2008, is a type of quantum teleportation protocol which transmits the state to the receiver without requiring any corrections on the receiver's side. Evaluating the performance of PBT was computationally intractable and previous attempts succeeded only with small systems. We study PBT protocols and fully characterize their performance for arbitrary dimensions and number of ports. We develop new mathematical tools to study the symmetries of the measurement operators that arise in these protocols and belong to the algebra of partially transposed permutation operators. First, we develop the representation theory of the mentioned algebra which provides an elegant way of understanding the properties of subsystems of a large system with general symmetries. In particular, we introduce the theory of the partially reduced irreducible representations which we use to obtain a simpler representation of the algebra of partially transposed permutation operators and thus explicitly determine the properties of any port-based teleportation scheme for fixed dimension in polynomial time.
\end{abstract}	

\section{Introduction}
Quantum teleportation is one of the most important primitives in the Quantum Information Processing~\cite{bennett_teleporting_1993}. This technique allows to transfer the state of an unknown quantum system from the sender to the receiver without having to exchange the physical system. It has led to a large number of theoretical advances in quantum information theory and quantum computing~\cite{boschi_experimental_1998,gottesman_demonstrating_1999,gross_novel_2007, jozsa_introduction_2005, knill_scheme_2001, pirandola_advances_2015, raussendorf_one-way_2001, zukowski_event-ready-detectors_1993}. 

The first teleportation protocol involved two parties, Alice and Bob, each sharing a half of the maximally entangled state~\cite{bennett_teleporting_1993}. We will further refer to it as a `resource state'.  Alice wants to send an (unknown) state of a subsystem in her possession to Bob. She performs a projective measurement on her subsystem and the half of the maximally entangled state and communicates its classical outcome to Bob. He then reliably recovers the state which Alice communicated by applying a unitary correction operation conditioned on Alice's message. 


In 2008, a breakthrough result from Ishizaka and Hiroshima introduced a novel port-based teleportation protocol (PBT) which does not require the last step in the sequence~\cite{ishizaka_asymptotic_2008}. In this setup, parties share a large resource state consisting of $N$ copies of the maximally entangled states $|\Psi^-\rangle^{\otimes N}$, where each singlet is a two-qubit state, termed {\it port}.  Alice performs a joint measurement $\cal X$ on the unknown state $\theta$ which she wishes to teleport and her half of the resource state, communicating the outcome to Bob. 
The outcome of the measurement points to the subsystem where the state has been teleported to. To obtain the teleported state Bob discards all ports except for the one indicated by Alice's outcome. There are two versions of the PBT protocol, depending on the exact set of measurements used by Alice. The first type, so-called {\it deterministic} teleportation, is described by the set of $N$ POVM elements ${\cal X} = \{\Pi_a\}_{a=1}^{N}$. 
Upon measuring $a$-th element the teleported state ends up in the $a$-th port on Bob's side. He then traces out all but $a$-th subsystem which contains the teleported state. The second type, ${\it probabilistic}$ PBT, consists of a measurement with $N+1$ POVM elements $\{\Pi_a\}_{a=0}^{N}$, where $\Pi_0$ indicates a failure of the teleportation. In this protocol, when Alice obtains the input $a\in\{1,\ldots,N\}$, the parties proceed as above. When she obtains $0$, then they abort the protocol. 

In the probabilistic PBT the state of a qubit always gets teleported to Bob, but it decoheres during the process. The lower bound on the fidelity of the teleported state tends to 1 as the number of ports $N\to\infty$. In the deterministic case protocol, the state always gets teleported to Bob with perfect fidelity, but with some probability (which vanishes in the limit $N\to\infty$) Alice aborts.

PBT schemes found novel applications in the areas where the existing teleportation schemes fell short of.  They provided new architecture for the universal programmable quantum processor performing computation by teleportation with the property of it being composable.~\cite{ishizaka_quantum_2009}. In position-based cryptography, PBT schemes were used to engineer efficient protocols for instantaneous implementation of measurement and computation. It resulted in new attacks on the cryptographic primitives, reducing the amount of consumable entanglement from doubly exponential to exponential~\cite{beigi_konig}. 

Recently, the composable nature of the qubit PBT schemes made it possible to connect the field of communication complexity and a Bell inequality violation~\cite{buhrman_quantum_2016}. It allowed to show that any quantum advantage obtained by a protocol for an arbitrary communication complexity problem resulted in the violation of a Bell inequality, certifying the quantum nature of the advantage. 

The full characterization of the qubit PBT schemes was used to obtain the performance of the square-root measurements for mixed states obtaining explicit probabilities of success when the set of states to be discriminated has certain symmetries and POVMs are of nearly maximal rank~\cite{beigi_konig}.

Evaluating the performance of the PBT is tantamount to determining the spectral properties of the measurement operators $\cal X$. To determine them, authors in~\cite{ishizaka_asymptotic_2008} viewed $N+1$ qubits (with one extra qubit representing the teleported state) as spins, recursively building a basis for constituents of $\cal X$ making the use of the Clebsch-Gordan (CG) coefficients and with the painstaking amount of effort determined their eigenvalues. This approach has been successful for studying systems of $N+1$ qubits and relied on the existence of the closed form for the CG coefficients and therefore was limited to $SU(2)^{\otimes N}$. In the case of $SU(d)^{\otimes N}$, with $d>2$ there exists no closed form of the CG coefficients and thus it is impossible to obtain the spectrum of $\cal X$ without incurring an exponential overhead in $d$ and $N$. It is however possible to obtain a closed-form lower bound on the performance of deterministic PBT, but it only works in the regime $N\gg d$. Moreover, no bound is known for the probabilistic PBT.

By using graphical variant of Temperley-Lieb algebra, authors obtained an explicit closed-form expressions for the fidelity and success probability of PBT for an arbitrary $d$ and $N\in \{2,3,4\}$~\cite{wang_higher-dimensional_2016}. Here too the mathematical formulae contain the number of different terms which grows exponentially in $N$.

In our work we develop new mathematical tools to study the symmetries of $\cal X$ which enable us to efficiently evaluate the performance of ${\it any}$ PBT scheme for arbitrary $N$ and $d$. Our first contribution is the theory of the partially reduced irreducible representations (PRIR). They provide an elegant way of understanding the properties of subsystems of a large system which has general symmetries. We further use these techniques to provide a simple way to approach to the representation of the algebra of the partially transposed permutation operators. Remarkably, the operators describing measurements in any PBT scheme possess the exact symmetries of an element of this algebra. By exploring these symmetries in a principled way we are able to fully analyze all teleportation schemes.

We thus characterize the performance of the main PBT schemes and find exact expressions for the fidelity of the teleportation and the probability of success in the deterministic and probabilistic schemes respectively. Moreover, we describe the spectral properties of the POVMs and exhibit polynomial algorithms to efficiently calculate the properties of quantum systems with similar symmetries using our framework. 


\section{Setting and main results}\label{Mainresults}

In this section we introduce the setting and outline the results obtained by new mathematical techniques developed in our work.

Consider a probabilistic protocol defined by a $d$-dimensional maximally entangled resource state $\bigotimes_{i=1}^N|\Phi^+_d\rangle_{{A_iB_i}}$, set of POVMs ${\cal X} = \{\Pi_a\}_{a=0}^{N}$, where each $\Pi_a = \rho^{-\frac{1}{2}}\varrho_a\rho^{-\frac{1}{2}}$ for $a\ge 1$ and $\Pi_0 = \mathbf{1} - \sum_{a=1}^N\Pi_a$ with 
\be\label{rho}
\rho=\sum_{a=1}^{N}\varrho_a,
\ee and $\varrho_a=\frac{1}{d^N}P^+_{CA_a}\ot \mathbf{1}_{\overline{A_a}}$, for  $a=1,\ldots, N$.

\begin{figure}[ht]
	\centering
	\includegraphics[width=0.5\textwidth]{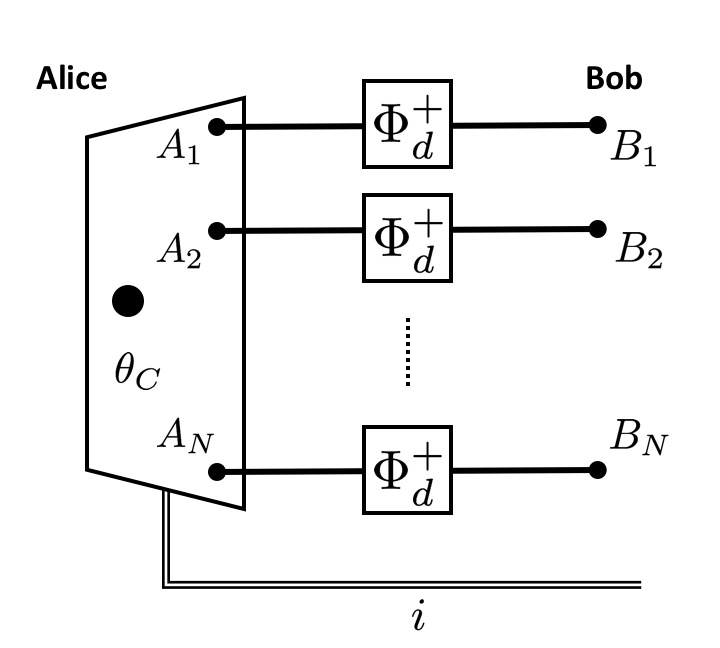}
	\caption{Schematic description of PBT in the arbitrary dimension.}
	\label{fig:pbt_general}
\end{figure}
The operator $P^+_{CA_a}$ denotes an unnormalised projection onto the state $|\Phi^+\rangle_{CA_a}=\sum_{i=1}^d|ii\>_{CA_a}$ between systems $C$ and $A_a$, and $\mathbf{1}_{\overline{A_a}}$ is identity operator on all subsystems $ A$ except $A_a$. {We will henceforth refer to $\rho$ as the PBT operator}.
Alice wishes to teleport a qudit $\theta_C$ to Bob. After she measures ${\cal X}$ and communicates the classical outcome $i$ to Bob, he performs one of the following actions: (a) if $i\in\{1,\ldots, N\}$ he traces out all but the $i$-th port which contains the teleported state with perfect fidelity; (b) if $i=0$ he aborts.
The schematic representation is shown in Fig~\ref{fig:pbt_general}.
\subsection{Eigenvalues of the PBT operator}
We develop new mathematical tools and find that the PBT operator has the following spectrum:

\begin{customprop}{2}{[Short version]}
	The eigenvalues of the PBT operator $\rho$ in Eqn.\eqref{rho} are given by $\lambda_{\mu}(\alpha)=\frac{N}{d^N}\frac{m_{\mu}d_{\alpha}}{m_{\alpha}d_{\mu}}$.
	Quantities $m_{\alpha},m_{\mu}$  denote multiplicities in the natural representation,  $d_{\alpha}, d_{\mu}$ denote the respective dimensions of the irreps. By $\mu$ we denote Young diagrams obtained from Young diagrams $\alpha$ of $n-2$ boxes by adding a single box in a proper way. We take Young diagrams $\alpha, \mu$ whose height is not greater then dimension $d$ of the local Hilbert space.
\end{customprop}

Further in this manuscript by symbol $\nu \vdash m$ we denote  Young diagram of $m$ boxes, by $\mu \in \alpha$ we denote all Young diagrams $\mu$ which can be obtained from the Young diagram $\alpha$ by adding a single box in a proper way, $\alpha \in \mu$ denotes all Young diagrams $\alpha$ which can be obtained from the Young diagram $\mu$ by removing a single box. By 'proper way' we understand a situation when the height of final Young diagram is less or equal than dimension $d$ of the local Hilbert space. For the Young diagram $\mu$ its height (number of rows) is denoted as $h(\mu)$. Since there is one-to-one correspondence between Young diagrams made up of $n$ boxes and inequivalent irreps of the symmetric group $S(n)$ we use symbols $\alpha, \mu$ etc. interchangeably for Young diagrams and irreps whenever it is clear from the context. 
\subsection{Probabilistic PBT}
In this scheme the teleported state reaches the recipient with high probability  and one distinguishes two different protocols each with perfect teleportation fidelity. In the first type we consider a resource state which consists of a number of maximally entangled states, and in the second type we obtain the resource state as well set of POVMs as a result of the optimization procedure. In the former case, the probability of success is given by

\begin{customthm}{3}
	The maximal average success probability in the probabilistic PBT with a resource state consisting of maximally entangled pairs is given by 
	${p= \frac{1}{d^N}\sum_{\alpha}m^2_{\alpha}\min_{\mu \in \alpha}\frac{d_{\mu}}{m_{\mu}}}$,
	where $\mu$ denotes Young diagram obtained from Young diagrams $\alpha \vdash n-2$  by adding a single box in a proper way and $\gamma_{\mu}(\alpha)$ is given in Proposition~\ref{m2}.
\end{customthm}

In the latter case, the probability of success is given by

\begin{customthm}{4}
    The optimal state in the probabilistic PBT is given by
	$X_A=\sum_{\mu}c_{\mu}P_{\mu}\quad \text{with}\quad c_{\mu}=\frac{d^{N}g(N)m_{\mu}}{d_{\mu}}$
	where $g(N)=1/\sum_{\nu}m_{\nu}^2$, and $\nu$ labels irreps of $S(n-1)$. Operators $P_{\mu}$ are Young projectors onto irreps of $S(n-1)$. The corresponding optimal probability is of the form $p=1-\frac{d^2-1}{N+d^2-1}$,
	where $N$ is the number of ports, and $d$ is the dimension of the local Hilbert space.
\end{customthm}

Figure~\ref{fig:psuccess} depicts the success probability computed by our algorithm for both cases.
\begin{figure}[ht]
	\centering
	\includegraphics[width=0.5\textwidth]{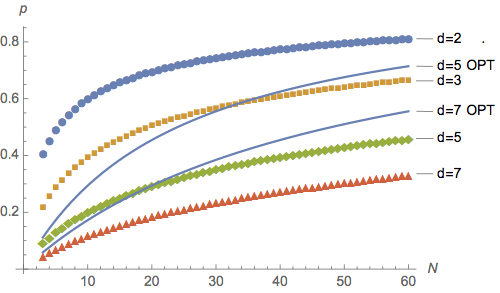}
	\caption{Exact performance of the probabilistic PBT protocol. Dotted lines correspond to the average success probability when we only optimize the measurements using maximally entangled resource state. Solid lines correspond to the average success probability when we optimize both the measurement and the resource state. }
	\label{fig:psuccess}
\end{figure}

\subsection{Deterministic PBT}
The deterministic version of the PBT guarantees that the teleported state always reaches the recipient, but at a cost of being distorted.
It is described by $N$ POVM elements $\{\Pi_a\}_{a=1}^{N}$ with each $\Pi_a = \rho^{-1/2}\varrho_a\rho^{-1/2} + \Delta$, where the term $\Delta=(1/N)(\mathbf{1} -\sum_{a=1}^N\Pi_a)$ with $\tr\varrho_a\Delta=0$ is required to ensure that $\sum_{a=1}^N\Pi_a = \mathbf{1}$. For simplicity, we take $\theta_C$ to be a half of the maximally entangled state. As described in the introduction, the sender, Alice, performs a joint measurement $\{\Pi_a\}_{a=1}^{N}$ on $\theta_C$ and her share of the resource state. She then communicates  the classical outcome $i\in \{1,\ldots, N\}$ to Bob, who then traces out all but the $i$-th port which contains the teleported state.

\begin{customthm}{12}
The fidelity for Port-Based Teleportation is given by $
F=\frac{1}{d^{N+2}}\sum_{\alpha \vdash n-2}\left(\sum_{\mu \in \alpha}\sqrt{d_{\mu}m_{\mu}} \right)^2$,
where sums over $\alpha$ and $\mu$ are taken, whenever number of rows in corresponding Young diagrams is not greater than  the dimension of the local Hilbert space $d$. 
\end{customthm}

From the entanglement fidelity computed in Theorem 12 one can easily obtain the average fidelity using $f= (Fd+ 1)/(d+ 1)$.
Figure~\ref{fig:entfidelity} shows the performance of the deterministic PBT when Alice wishes to teleport higher-dimensional states.
\begin{figure}[ht]
	\centering
	\includegraphics[width=0.5\textwidth]{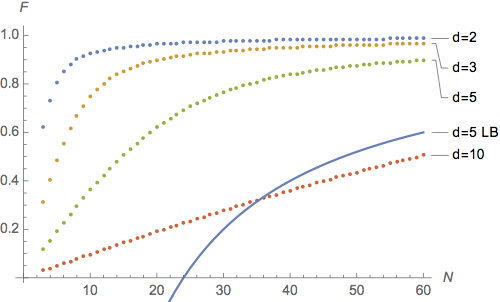}
	\caption{Performance of the deterministic PBT protocol for $d\in\{2,3,4,10\}$. Dotted line denotes explicit values for the entanglement fidelity computed by our algorithm. Solid line denotes the best lower bound for $d=5$ derived in~\cite{ishizaka_asymptotic_2008}.}
	\label{fig:entfidelity}
\end{figure}

\section*{Discussion and open questions}\label{Discussion}
We found explicit expressions for the performance of all variations of the PBT in arbitrary dimension with any number of ports. We expect the tools and techniques introduced here to find a number of applications ranging from the study of quantum states with restricted symmetries and calculating properties of the antiferromagnetic systems to problems in the quantum measurement theory. 

{Our successful approach to studying properties of the port-based teleportation may be replicated for the study of arbitrary systems with partial symmetries. First, one needs to classify the symmetries of the system ($S(n-1)$ and $S(n-2)$ in the case of the PBT). Next, to identify and study the structure and the natural representation of the algebra corresponding to the elements with these symmetries ($\mathcal{A}^{t_n}_n(d)$ in the case of the PBT). Finally, to compute various tracial quantities of interest efficiently one needs to adjust the theory PRIRs to account for the symmetries of the system in question.}

We now mention some open questions. Firstly, how to characterize entanglement content in the resource state after one run of any PBT protocol for $d>2$. We know that in the qubit case the residual entanglement may be recycled to teleport more states~\cite{strelchuk_generalized_2013}. In addition, when probabilistic PBT fails, Alice can nevertheless make use of the standard teleportation protocol reliably~\cite{ishizaka_remarks_2015}. One might wonder if it is possible to similarly utilize the residual entanglement of the resource state in the qudit case. In particular, to show that such higher-dimensional teleportation scheme works one needs to extend our analysis to the properties of the left ideal $\cal S$ depicted in Fig~\ref{fig:algebraStruct}.

Second, for the probabilistic PBT protocol we have shown that the measurement operators are optimal for a fixed resource state of the form $|\Phi^+\rangle^{\otimes N}$. We do not know whether this holds for the deterministic PBT in our case. From~\cite{grudka_optimal_2008, ishizaka_asymptotic_2008, ishizaka_quantum_2009} we know that for both the KLM scheme and the PBT protocols teleporting qubits, the resource state and the corresponding measurement differ when optimized simultaneously.

Another important problem is to determine the asymptotic performance of the PBT in arbitrary dimension. This presents a challenging task in particular because the asymptotic representation theory for the regime $d/N\to0$ is still in its infancy.

\section*{Methods}
\subsection{Structure of the port-based teleportation operator}
\label{Spectrum}
 In order to quantify the effect of measurement $\cal X$ one has to find the spectral properties of $\rho$. To simplify the analysis, we represent $\rho$ in a different form. First, observe that every unnormalised projector $P^+_{CA_a}$ can be written as $P^+_{CA_a}=V^{t_{C}}_{(CA_a)}$,
where $V_{(CA_a)}$ denotes a permutation operator  between systems $C$ and $A_a$, and $t_{C}$ denotes a transposition with respect to subsystem $C$.  Therefore, $\rho$ in Eqn.~\eqref{rho} may be written in terms of partially transposed permutation operators:
\begin{equation}
\label{introPswap}
\rho=\frac{1}{d^N}\sum_{a=1}^{N}V^{t_{C}}_{(CA_a)}\ot \mathbf{1}_{\overline{A_a}}.
\end{equation}
Every element $V^{t_{C}}_{(CA_a)}\ot \mathbf{1}_{\overline{A_a}}$ acts as a permutation operator on a full $n=N+1$-particle (we will use $n-1$ and $N$ interchangeably) Hilbert space $\mathcal{H}=\left( \mathbb{C}^{d}\right)^{\ot n}$. When it is clear from the context, we denote every operator  $V^{t_{C}}_{(CA_a)}\ot \mathbf{1}_{\overline{A_a}}$ just by $V^{t_{C}}_{(CA_a)}$.
The above form enables us to identify $\rho$ as the element of a recently studied algebra of partially transposed permutation operators $\mathcal{A}_{n}^{\operatorname{t}_{n}}(d)$ acting  in the space $\left( \mathbb{C}^{d}\right)^{\ot n}$, where $d\in \mathbb{N}$ and $d\geq 2$~\cite{Moz1,Stu1}. It turns out that $\mathcal{A}^{t_n}_n(d)$ decomposes into a direct sum of two types of left ideals (A left ideal of an algebra $\mathcal{A}$ is a subalgebra $\mathcal{I}\subset \mathcal{A}$ such that $ax\in \mathcal{I}$ whenever $a\in \mathcal{A}$ and $x\in \mathcal{I}$) $\mathcal{A}^{t_n}_n(d)=\mathcal{M}\oplus \mathcal{S}$.
To describe the functioning of the PBT protocols it suffices to only consider $\mathcal{M}$. The latter includes the irreducible representations of $\mathcal{A}_{n}^{t_{n}}(d)$  indexed by the irreducible
representations of the group $S(n-2)$ which are strictly connected with
the representations of the group $S(n-1)$~\cite{Moz1,Stu1} (see Fig.~\ref{fig:algebraStruct}).
\begin{figure}[ht]
	\centering
	\includegraphics[width=0.5\textwidth]{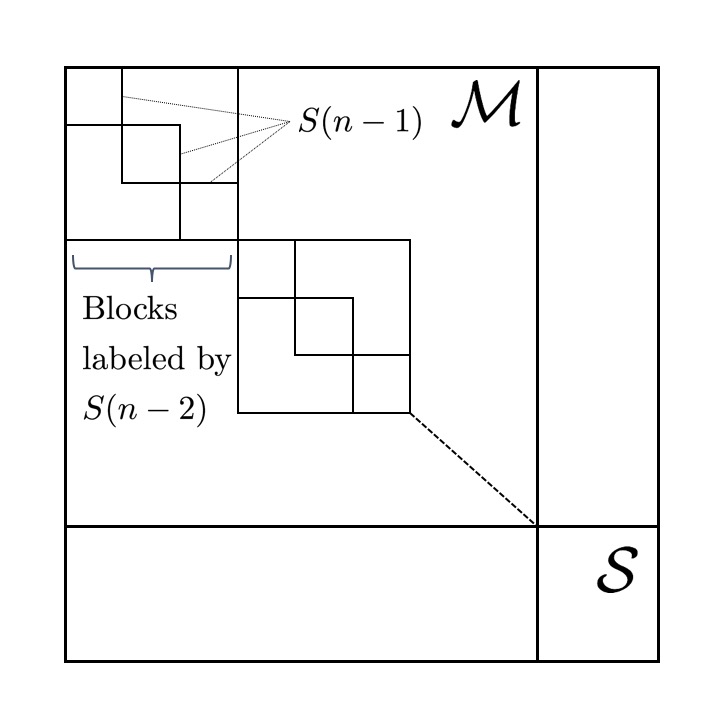}
	\caption{The structure of $\mathcal{A}_{n}^{\operatorname{t}_{n}}(d)$. It splits into direct sum of two ideals $\mathcal{M}$ and $\mathcal{S}$. The irreps of $\mathcal{M}$ are labelled by the irreps of $S(n-2)$ and they are strictly connected with the representations of the group $S(n-1)$ induced from the irreps of $S(n-2)$~\cite{Moz1}.}
	\label{fig:algebraStruct}
\end{figure}
To keep the notation consistent with the previous analysis of the algebra $\mathcal{A}_{n}^{\operatorname{t}_{n}}(d)$, we consider operator $\rho$ without factor $1/d^{n-1}$ and we change the numbering of the subsystems rewriting the general form of equation~\eqref{introPswap} as:
\be
\label{ro1}
\eta=\sum_{a=1}^{n-1}V^{t_{n}}(a,n),
\ee
where $t_n$ denotes a partial transposition on the $n^{\text{th}}$ subsystem, and $(a,n)$ is a permutation between subsystems $a$ and $n$. Here, the subsystem $C$ is labelled by $n$.
To evaluate the performance of the PBT scheme explicitly, we need to characterize the spectral properties of the above operator given in~\eqref{ro1}. For $d=2$ and $N\ge2$ this was done in~\cite{ishizaka_asymptotic_2008} using the CG coefficients. The constraint for the dimension cannot be improved because the CG formalism does not admit closed-form solutions beyond $d=2$. Our first contribution is the operator decomposition of $\eta$ which leads to a universal decomposition method which works for all $d\ge2$ and  $N\ge2$.

\begin{theorem}
	\label{m1}
	The operator $\eta=\sum_{a=1}^{n-1}V^{t_n}(a,n)$ has the form:
	\be
	\label{m1eq1}
	\eta=\bigoplus_{\alpha \vdash n-2}\bigoplus_{\substack{\mu \vdash n-1 \\ \mu \in \alpha}}\eta_{\mu}(\alpha)=\bigoplus_{\alpha \vdash n-2}\bigoplus_{\substack{\mu \vdash n-1 \\ \mu \in \alpha}} P_{\mu}\sum_{a=1}^{n-1}V(a,n-1)P_{\alpha} V^{t_n}(n-1,n)V(a,n-1),
	\ee
	where $\alpha,\mu$ run over Young diagrams whose heights are not greater then  the dimension $d$ of a single system, $\mu \in \alpha$ denotes a valid Young diagram $\mu$ obtained from $\alpha$ by adding one box in a proper way, $P_{\mu}, P_{\alpha}$ denote Young projectors onto irreps of $S(n-1),S(n-2)$ labelled by $\mu \vdash n-1, \alpha\vdash n-2$ respectively. Each $\eta_{\mu}(\alpha)$ is proportional to a projector, i.e. $\eta_{\mu}(\alpha)=\gamma_{\mu}(\alpha)F_{\mu}(\alpha)$, $\gamma_{\mu}(\alpha)\in \mathbb{C}$, $F_{\mu}(\alpha)$ is a projector of the dimension $\dim F_{\mu}(\alpha)=d_{\mu}\widetilde{m}_{\alpha}$, where $\widetilde{m}_{\alpha}$ is the multiplicity of the {irrep of the algebra labelled by $\alpha$} in $\mathcal{A}_{n}^{\operatorname{t}_{n}}(d)$, and $d_{\mu}$ is the dimension of the irrep of $S(n-1)$ labelled by $\mu$. The projectors $F_{\mu}(\alpha)$ satisfy  $F_{\mu}(\alpha)=M_{\alpha}P_{\mu}$, where $M_{\alpha}$ is projector including multiplicities onto $\alpha$-th irrep of the algebra $\mathcal{A}_{n}^{\operatorname{t}_{n}}(d)$.
\end{theorem}
\noindent{The structure and relation of projectors that appear in the theorem are depicted in Figure~\ref{l3}.}
\begin{proof}
	Fix an arbitrary representation of algebra $\mathcal{A}_{n}^{\operatorname{t}_{n}}(d)$. Let $M_{\alpha}$ be projector (including multiplicities) onto irrep labelled by $\alpha \vdash n-2$; denote the corresponding subspace $S(M_{\alpha})$, and set $P_{\mu}$ to be a projector onto irrep of $S(n-1)$ in the same representation (including the multiplicities). Our first goal is to determine the restriction of the operator $\eta$ to the irrep labelled by $\alpha$. We express $\eta$ in terms of operators $\{v_{ij}^{ab}(\alpha)\}$, where $1\leq a,b \leq n-1$ and $1\leq i,j \leq d_{\alpha}$ (see Definition 5 in~\cite{Stu1})
	\be
	\label{ppp}
	v_{ij}^{ab}(\alpha)=V(a,n-1)E_{ij}^{\alpha} V^{t_n}(n-1,n)V(b,n-1),
	\ee
	since they span irrep labelled by $\alpha$. The operators $\{E_{ij}^{\alpha}\}_{i,j=1}^{d_{\alpha}}$ form an operator basis in the irrep $\alpha \vdash n-2$ of $S(n-2)$ (see Appendix~\ref{A2} for the  details). Using ~\eqref{ppp} we can decompose $\eta$:
	\be
	\begin{split}
		\eta&=\sum_{a=1}^{n-1}V(a,n-1)V^{t_n}(n-1,n)V(a,n-1)=\sum_{\alpha}\sum_{a=1}^{n-1}V(a,n-1)P_{\alpha}V^{t_n}(n-1,n)V(a,n-1)\\
		&=\sum_{\alpha}\sum_{i=1}^{d_{\alpha}}\sum_{a=1}^{n-1}V(a,n-1)E_{ii}^{\alpha}V^{t_n}(n-1,n)V(a,n-1)=\sum_{\alpha}\sum_{i=1}^{d_{\alpha}}\sum_{a=1}^{n-1} v_{ii}^{aa}(\alpha)=\sum_{\alpha}\eta(\alpha),
	\end{split}
	\ee
where 
\be
\eta(\alpha)=\sum_{i=1}^{d_{\alpha}}v_{ii}^{aa}(\alpha)=\sum_{a=1}^{n-1}V(a,n-1)P_{\alpha}V^{t_n}(n-1,n)V(a,n-1).
\ee	
Thus the support of $\eta(\alpha)$ is precisely the space $S(M_{\alpha})$ which is invariant under the action of $S(n-1)$, hence its eigenprojectors are $F_{\mu}(\alpha)=M_{\alpha}P_{\mu}$, and this results in the following decomposition:
	\be
	\label{eq8}
	\eta(\alpha)=\bigoplus_{\substack{\mu \vdash n-1 \\ \mu \in \alpha}}\gamma_{\mu}(\alpha)M_{\alpha}P_{\mu}=\bigoplus_{\mu \in \alpha}\eta_{\mu}(\alpha),\quad \gamma_{\mu}(\alpha)\in\mathbb{C},
	\ee
	with $\eta_{\mu}(\alpha)=P_{\mu}\eta(\alpha)P_{\mu}$. This immediately implies that
	\be
	\label{FF}
	F_{\mu}(\alpha)=\gamma^{-1}_{\mu}(\alpha)P_{\mu}\eta(\alpha)P_{\mu}.
	\ee
	All of the structural properties above are derived solely from the properties of the underlying algebra, and are thus independent of representation. It is known that for any representation $\tr\left[P_{\mu}M_{\alpha} \right]=d_{\mu}\widetilde{m}_{\alpha}$, where $\widetilde{m}_{\alpha}$ is the multiplicity of the projector $M_{\alpha}$~\cite{Curtis}. 
\end{proof}
\begin{figure}[ht]
	\centering
	\includegraphics[width=0.7\textwidth]{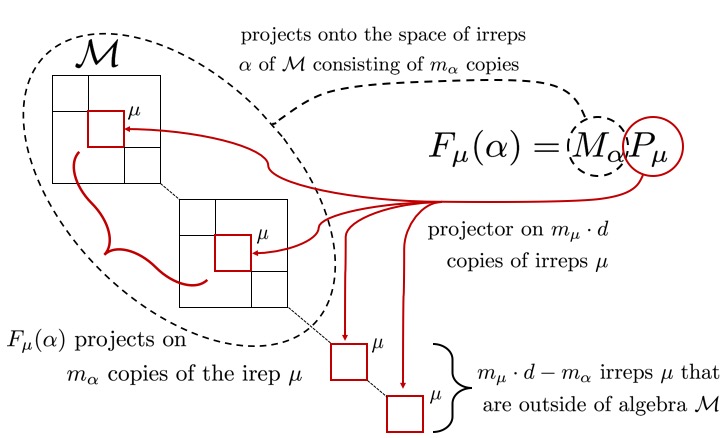}
	\caption{Graphical illustration of the action of the projector $F_{\mu}(\alpha)$ in  Theorem~\ref{m1}.}
	\label{l3}
\end{figure}
To simplify the presentation, we may occasionally switch to the natural representation (for instance when we want to compute the partial trace). A number of fundamental results about the structure of the above algebra was obtained by~\cite{Moz1,Stu1} who in particular showed that $\widetilde{m}_{\alpha}=m_{\alpha}$, where $m_{\alpha}$ is the multiplicity of irrep labelled by $\alpha$ in the natural representation of the group $S(n-2)$.
Keeping all the notation introduced in the previous theorem we now find the formula for the eigenvalues of the operator $\eta$ as well as port-based teleportation operator $\rho$. 
\begin{proposition}{[Extended version]}
	\label{m2}
	The numbers $\gamma_{\mu}(\alpha)$ given by Eqn.~\eqref{eq8} are the eigenvalues of the operator $\eta$ given by
	\be
	\label{m2eq1}
	\gamma_{\mu}(\alpha)=(n-1)\frac{m_{\mu}d_{\alpha}}{m_{\alpha}d_{\mu}}
	\ee
	or equivalently: 
	\be
	\label{m2eq2}
	\gamma_{\mu}(\alpha)=d+\frac{1}{2}(n-1)(n-2)%
	\frac{\chi ^{\mu }(12)}{d_{\mu}}-\frac{1}{2}(n-2)(n-3)\frac{%
		\chi ^{\alpha }(12)}{d_{\alpha}}.
	\ee
	By $m_{\alpha},m_{\mu}$ we denote multiplicities of {$\alpha, \mu$ of $S(n-2), S(n-1)$ respectively} in the natural representation, by $d_{\alpha}, d_{\mu}$ the respective dimensions, and by $\chi^{\mu}(12), \chi^{\alpha}(12)$ the characters calculated on the transposition $(12)$ of the corresponding irreps. By $\mu$ we denote Young diagrams obtained from Young diagrams $\alpha$ of $n-2$ boxes by adding a single box in a proper way. We take Young diagrams $\alpha, \mu$ whose height is not greater then dimension $d$ of the local Hilbert space.
\end{proposition}

\begin{proof}
From Theorem~\ref{m1} we know that $\eta=\sum_{\substack{\mu \vdash n-1 \\ \mu \in \alpha}} \gamma_{\mu}(\alpha)F_{\mu}(\alpha)$, and $\eta_{\mu}(\alpha)=\gamma_{\mu}(\alpha)F_{\mu}(\alpha)$. Thus $\gamma_{\mu}(\alpha)$ can be expressed as
\be
\label{gamma}
\gamma_{\mu}(\alpha)=\frac{\tr \eta_{\mu}(\alpha)}{\tr F_{\mu}(\alpha)}=\frac{\tr \eta_{\mu}(\alpha)}{d_{\mu}m_{\alpha}}.
\ee 
In the last step we need to compute $\tr \eta_{\mu}(\alpha)$. Using the decomposition given in equation~\eqref{m1eq1} and Fact~\ref{fact11} from Appendix~\ref{AppC} we can simplify:
\be
\label{3}
\tr \eta_{\mu}(\alpha)=(n-1)\tr\left[P_{\mu}P_{\alpha}V^{t_n}(n-1,n) \right]=\tr\left[P_{\mu}\left(P_{\alpha}\ot \mathbf{1} \right)  \right]=(n-1)m_{\mu}d_{\alpha}.
\ee
Plugging~\eqref{3} into~\eqref{gamma} we obtain first statement of the proposition given in equation~\eqref{m2eq1}.\\\
To prove~\eqref{m2eq2} we {use~\eqref{gamma} and~\eqref{3}} express $\gamma_{\mu}(\alpha)$ as:
\be
\gamma_{\mu}(\alpha)=\frac{n-1}{d_{\mu}m_{\alpha}}\tr\left[P_{\mu}\left( P_{\alpha}\ot \mathbf{1}\right) \right].
\ee
Using Fact~\ref{Fij} from Appendix~\ref{AppC} to obtain the decomposition of $P_{\mu}$ and simplifying it further we get:
\be
\label{eq15}
\begin{split}
	\gamma_{\mu}(\alpha)&=\frac{n-1}{m_{\alpha}(n-1)!}\sum_{a=1}^{n-1}\sum_{i,j=1}^{d_{\mu}}\varphi_{ij}^{\mu}(a,n-1)\tr\left[V(a,n-1)F_{ij}^{\mu} \left(P_{\alpha}\ot \mathbf{1} \right) \right]\\
	&=\frac{n-1}{m_{\alpha}(n-1)!}\sum_{a=1}^{n-1}\sum_{i,j=1}^{d_{\mu}}d^{\delta_{a,n-1}}\varphi_{ij}^{\mu}(a,n-1)\tr\left[F_{ij}^{\mu}P_{\alpha} \right],
\end{split}
\ee
{where \be\label{Fij}F_{ij}^{\mu}=\sum_{\pi \in S(n-2)} \varphi_{ji}^{\mu}\left(\pi^{-1} \right)V(\pi).\ee}
From Eqn.~\eqref{Fprop} in Fact~\ref{Fij} in Appendix~\ref{AppC} and the orthogonality property that  $\tr\left[E_{ij}^{\beta}P_{\alpha} \right]=\sum_{k}\tr\left[E_{ij}^{\beta}E_{kk}^{\alpha} \right]=\tr\left[ E_{ij}^{\alpha}\right]=\delta_{ij}m_{\alpha}$, we write~\eqref{eq15} in the PRIR notation defined in Appendix~\ref{AppA}
\be
\begin{split}
	\gamma_{\mu}(\alpha)=\frac{1}{d_{\alpha}}\sum_{a=1}^{n-1}d^{\delta_{a,n-1}}\left( \sum_{i_{\alpha}=1}^{d_{\alpha}}\left( \varphi^{\mu}_R\right) ^{\alpha \alpha}_{i_{\alpha}i_{\alpha}}(a,n-1)\right).
\end{split}
\ee
Finally, using Corollary~\ref{Cc} from Appendix~\ref{AppA} we obtain the second statement of the proposition.
\end{proof}
{Thus, the eigenvalues of the PBT operator $\rho$ given in Eqn.~\eqref{rho} are the rescaled version of the above: {\be
\lambda_\mu(\alpha)=\left( 1/d^{N}\right) \gamma_\mu(\alpha) = \left( N/d^N\right) \frac{m_{\mu}d_{\alpha}}{m_{\alpha}d_{\mu}}.
\ee}}
\subsection{Probabilistic version of the protocol}\label{Probabilistic}
In the following two subsections we find maximal average success probability when a resource state is the maximally entangled state and then  show how to find optimal resource state and POVMs simultaneously.

To prove the optimality in the above cases, we formulate the question as a semidefinite program. We prove the main theorems by {presenting feasible solutions for a primal and a dual} semidefinite problem. 
By establishing the solution to a primal problem we obtain an achievable lower bound for the success probability; the corresponding solution to the dual yields the upper bound. Observing that the respective bounds coincide we arrive at the optimal probability of success $p_{opt}$.

\subsubsection{Maximally entangled state as a resource state}
\label{Prob:max}
 From~\cite{beigi_konig} it follows that the optimal POVMs for the probabilistic PBT coincides with the ones for distinguishing the set of states $\{(1/N;\varrho_a)\}_{a=1}^N$. We thus look for a set of POVMs $\{\Pi_a = P^+_{a,n}\ot \Theta_{\overline{a}}\}_{a=1}^N$ which would maximize the average success probability 
\be\label{opt}
p^\star=\frac{1}{d^{N+1}}\sum_{a=1}^N\tr \Pi_a = \frac{1}{d^{N+1}}\sum_{a=1}^N\tr\Theta_{\overline{a}}
\ee
subject to:
\be
\label{cons}
\mbox{(1) }\Theta_{\overline{a}}\geq 0,\quad \mbox{(2) }\sum_{a=1}^{N} P^+_{a,n}\ot \Theta_{\overline{a}}\leq  \mathbf{1}_{AB}\quad a=1,2,\ldots,N.
\ee
Since our resource state is maximally entangled, the RHS of the second constraint in~\eqref{cons} reduces to identity on $AB$ (see~\cite{ishizaka_asymptotic_2008}). Our main contribution here is the explicit expression for the probability of success of PBT:
\begin{theorem}
	\label{m3}
	The maximal average success probability in the probabilistic PBT with a resource state consisting of maximally entangled pairs is given by 
	${p_{opt}= \frac{1}{d^N}\sum_{\alpha}m^2_{\alpha}\min_{\mu \in \alpha}\frac{d_{\mu}}{m_{\mu}}}$,
	where $\mu$ denotes Young diagram obtained from $\alpha \vdash n-2$ by adding a single box in a proper way and $\gamma_{\mu}(\alpha)$ is given in Proposition~\ref{m2}.
\end{theorem}

\begin{lemma}
	\label{primal_max}
	The primal is feasible with ${p^\star= \frac{1}{d^N}\sum_{\alpha}\min_{\mu \in \alpha}\frac{m_{\alpha}d_\alpha }{\gamma_{\mu}(\alpha)}}\le p_{opt}$.
\end{lemma}
The proof is located in Appendix~\ref{SDPlemmas.1}.
The dual problem  is to minimize
\be
\label{low_star}
p_{\star}=\frac{1}{d^{N+1}}\tr \Omega
\ee
subject to:
\be \label{dualA}\quad \mbox{(1) }\Omega \geq 0,\quad \mbox{(2) }\tr_{a,n}\left[P^+_{a,n}\Omega\right] \geq \mathbf{1},
\ee
where $a=1,\ldots, n-1$, $\Omega$ acts on $n$ systems the identity $\mathbf{1}$ is defined on $n-2$ systems, and $P^+_{a,n}$ is projector onto maximally entangled state between respective subsystems. We {choose the operator $\Omega$ by a linear combination} of the projectors $F_{\mu}(\alpha)$ defined in Theorem~\ref{m1}
\be
\label{om12}
\Omega=\sum_{\alpha}x_{\mu^*}(\alpha)F_{\mu*}(\alpha),\quad x_{\mu^*}(\alpha)\ge0,
\ee
where $\mu^* \vdash n-1$ denotes the Young diagram obtained from $\alpha \vdash n-2$ by adding one box in a proper way in such a way that $\gamma_{\mu^*}(\alpha)$ is possible maximal. From the definition of the constraints~\eqref{low_star} and the symmetries of the projectors $F_{\mu}(\alpha)$ we need the following fact

\begin{fact}
	\label{Fop}
	Let $F_{\mu}(\alpha)$ be the operators given in Theorem~\ref{m1}, and let $V^{t_n}(n-1,n)$ be a  permutation operator acting between $(n-1)$-th and $n$-th subsystems partially transposed with respect to $n$-th subsystem, then 
	$\tr_{n-1,n}\left[V^{t_n}(n-1,n) F_{\mu}(\alpha)\right]=\frac{m_{\mu}}{m_{\alpha}}P_{\alpha}$,
	where numbers $m_{\alpha},m_{\mu}$ are the multiplicities of the respective irreps and $P_{\alpha}$ is the Young projector onto a irreducible subspace labelled by the partition $\alpha \vdash n-2$.
\end{fact}
The proof of Fact~\ref{Fop} an the lemma below are located in Apendix~\ref{SDPlemmas.2} and~\ref{SDPlemmas.3} respectively.

\begin{lemma}
	\label{dual_max}
	The dual is feasible with $p_\star = \frac{1}{d^N}\sum_{\alpha}m_{\alpha} ^2\frac{d_{\mu^{\ast}}}{m_{\mu^{\ast}}}\ge p_{opt}$.
\end{lemma}
Combining Lemma~\ref{primal_max} with Lemma~\ref{dual_max} we formulate the following proposition:
\begin{proposition}
	From $p^\ast = p_\star$ we conclude that ${\Theta}_{\overline{a}}=d\sum_{\alpha}\frac{1}{\gamma_{\mu^*}(\alpha)}P_{\alpha}={d\sum_{\alpha}P_{\alpha}\min_{\mu \in \alpha}\frac{1}{\gamma_{\mu}(\alpha)}}$ for $a=1,2,\ldots,N$ are the optimal POVMs for the maximally entangled state as a resource state.
\end{proposition}

\subsubsection{Optimisation over a resource state}
\label{Prob:opt}
We now turn to the case when both the optimal POVMs and a resource state are optimized simultaneously. We thus look for a set of POVMs $\{\Pi_a = P^+_{a,n}\ot \Theta_{\overline{a}}\}_{a=1}^N$ which would maximize the average success probability
\be
\label{pres}
p^\star=\frac{1}{d^{N+1}}\sum_{a=1}^N\tr \Pi_a = \frac{1}{d^{N+1}}\sum_{a=1}^N\tr\Theta_{\overline{a}}
\ee
subject to
\be
\label{cons1}
\mbox{(1) }\Theta_{\overline{a}}\geq 0,
\quad \mbox{(2) }\sum_{a=1}^{N} P^+_{a,n}\ot \Theta_{\overline{a}}\leq  X_A\ot \mathbf{1}_{B}\quad a=1,2,\ldots,N.
\ee
In the above $X_A=O_A^{\dagger}O_A \geq 0$, where $O_A$ is an operation applied by Alice on her half of the resource state satisfying $\tr X_A=d^N$ (see~\cite{ishizaka_asymptotic_2008}). Using our formalism, we derive the optimal state for the probabilistic PBT:
\begin{theorem}
	\label{main:state}
	The optimal state in the probabilistic PBT is given by, $|\zeta_{opt}\rangle = (O_A\otimes\mathbf{1})|\psi^+\rangle^{\otimes N}$, where $O_A = X_A^{1/2}$ with 
	\be
	X_A=\sum_{\mu}c_{\mu}P_{\mu}\quad \text{with}\quad c_{\mu}=\frac{d^{N}g(N)m_{\mu}}{d_{\mu}},
	\ee
	where $g(N)=1/\sum_{\nu}m_{\nu}^2$, and $\nu$ labels irreps of $S(n-1)$. Operators $P_{\mu}$ are Young projectors onto irreps of $S(n-1)$. 
	The optimal set of POVMs  in the probabilistic PBT are given by $\{\Pi_a = P^+_{a,n}\ot \Theta_{\overline{a}}\}_{a=1}^N$, where
	\be
	\label{povmres2}
	\forall \ a=1,\ldots,N\quad \Theta_{\overline{a}}=\sum_{\alpha}u(\alpha)P_{\alpha,\overline{a}},\quad \text{with} \quad u(\alpha)=\frac{d^{N+1}g(N)m_{\alpha}}{Nd_{\alpha}}.
	\ee
	Above sum runs over all allowed irreps of $S(n-2)$, $g(N)=1/\sum_{\nu}m_{\nu}^2$ for all  $\nu \vdash n-1$. By $P_{\alpha, \overline{a}}$ we denote Young projectors onto irreps of $S(n-2)$ defined on every subsystem except $n$-th and $a$-th.
	The corresponding optimal probability is of the form $p_{opt}=1-\frac{d^2-1}{N+d^2-1}$, where $N$ is the number of ports, and $d$ is the dimension of the local Hilbert space.
\end{theorem}
We prove the above theorem by {presenting feasible solutions to a primal (Lemma~\ref{optprim}), auxiliary lemma (Lemma~\ref{Fop2}) and dual semidefinite problem (Lemma~\ref{dual2})}. Proofs can be found in Appendix~\ref{SDPlemmas.4},~\ref{SDPlemmas.5} and~\ref{SDPlemmas.6} respectively.
Defining again a feasible value of the primal problem as in~\eqref{opt}, but with respect to constraints~\eqref{cons1} we have
\begin{lemma}
	\label{optprim}
	The primal is feasible with 
	$p^{\star}=1-\frac{d^2-1}{N+d^2-1}$,
	where $N$ is the number of ports, and $d$ is the dimension of the local Hilbert space.
\end{lemma}

The dual problem is of minimizing $p_{\star}=d^Nb$, where $b\in \mathbb{R}_+$  subject to
\be
\label{cons2}
\quad \mbox{(1) }\Omega \geq 0,
\quad\mbox{(2) } \tr_{a,n}\left[P^+_{a,n}\Omega \right]\geq \mathbf{1}_{1\ldots n-2},
\quad \mbox{(3) }b\mathbf{1}_{1\ldots n-1}-\frac{1}{d^{N+1}}\tr_a\Omega \geq 0,
\ee
where $a=1,\ldots, n-1$, $\Omega$ acts on $n$ systems, identities $\mathbf{1}_{1\ldots n-2}, \mathbf{1}_{1\ldots n-1}$ are defined on $n-2$ and $n-1$ systems respectively, and $P^+_{a,n}$ is a projector onto maximally entangled state between respective subsystems. As in the case of the maximally entangled state we assume the general form of $\Omega$ to be given as a linear combination of the projectors $F_{\mu}(\alpha)$ defined in Theorem~\ref{m1}
\be
\label{om1}
\Omega=\sum_{\alpha}\sum_{\mu\in \alpha}x_{\mu}(\alpha)F_{\mu}(\alpha),\quad x_{\mu}(\alpha)\ge0,
\ee
where $\mu \vdash n-1$ denotes the Young diagram obtained from the Young diagram $\alpha \vdash n-2$ by adding one box in a proper way. From the definition of the constraints~\eqref{cons2} and the symmetries of the projectors $F_{\mu}(\alpha)$ we calculate its partial trace in terms of $P_\mu$:
\begin{lemma}
	\label{Fop2}
Let $F_{\mu}(\alpha)$ be the operators given in Theorem~\ref{m1}, then
$\tr_{n} F_{\mu}(\alpha)=\frac{m_{\alpha}}{m_{\mu}}P_{\mu}$,
where numbers $m_{\alpha},m_{\mu}$ are multiplicities of the respective irreps and $P_{\mu}$ is the Young projector onto irreducible subspace labelled by the partition $\mu \vdash n-1$.
\end{lemma}
Having Lemma~\ref{Fop2} we are ready to formulate the following:
\begin{lemma}
	\label{dual2}
The dual is feasible with 
$p_{\star}=1-\frac{d^2-1}{N+d^2-1}$,
where $N$ is the number of ports, and $d$ is the dimension of the local Hilbert space.
\end{lemma}

By combining Lemma~\ref{optprim} and Lemma~\ref{dual2} we find that  $p^\star = p_\star$, so we conclude that $\Theta_{\overline{a}}$ for $a=1,2,\ldots,N$ given in ~\eqref{povmres} are the optimal POVMs with the optimal state given by~\eqref{optstate}, concluding the proof of Theorem~\ref{main:state}.

One important consequence of Theorem~\ref{main:state} is that the optimal probability of success in probabilistic performance of port-based teleportation is given only in terms of `global' parameters such as the number of ports $N$ and the dimension of a local Hilbert space $d$, $p\equiv p(d,N)$. Thus, for any fixed dimension $d$ we get $\mathop{\operatorname{lim}}_{N \rightarrow \infty}p(d,N)=1$,
which shows that in the asymptotic limit our protocol achieves a unit success probability. Moreover, for fixed number of ports $N$ we have $\mathop{\operatorname{lim}}_{d \rightarrow \infty}p(d,N)=0$ as we expected.
Since the probability of success given in in Theorem~\ref{main:state} is optimal and upper bounds the probability of success in Theorem~\ref{m3} in the case of maximally entangled state for any value of $d$ and $N$.

\subsection{Deterministic version of the protocol}\label{Deterministic}
The deterministic version of the PBT is described by $N$ POVM elements $\{\Pi_a\}_{a=1}^{N}$ with each $\Pi_a = \rho^{-1/2}\varrho_a\rho^{-1/2} + \Delta$, where the term $\Delta=(1/N)(\mathbf{1} -\sum_{a=1}^N\Pi_a)$ with $\tr\varrho_a\Delta=0$ is required to ensure that $\sum_{a=1}^N\Pi_a = \mathbf{1}$. For simplicity, we take $\theta_C$ to be a half of the maximally entangled state. As described in the introduction, the sender, Alice, performs a joint measurement $\{\Pi_a\}_{a=1}^{N}$ on $\theta_C$ and her share of the resource state.
The entanglement fidelity of the protocol for any $d\ge2$, $N\ge2$ is given by:
 \begin{theorem}
	\label{main2}
	The fidelity for Port-Based Teleportation is given by the following formula $F=\frac{1}{d^{N+2}}\sum_{\alpha \vdash n-2}\left(\sum_{\mu \in \alpha}\sqrt{d_{\mu}m_{\mu}} \right)^2$,
	where sums over $\alpha$ and $\mu$ are taken, whenever number of rows in corresponding Young diagrams is not greater than  the dimension of the local Hilbert space $d$. 
\end{theorem}

\begin{proof}
	From~\cite{ishizaka_asymptotic_2008} we know that the fidelity in PBT is given by
	\be
	\label{p1}
	F=\frac{1}{d^{2}}\sum_{a=1}^N\tr\left[\varrho_a\rho^{-1/2}\varrho_a\rho^{-1/2} \right]=\frac{N}{d^{N+2}}\tr\left[V^{t_{n}}(n-1,n)\eta^{-1/2}V^{t_{n}}(n-1,n)\eta^{-1/2} \right],
	\ee
where the right-hand side is recast using our notation. We use fact that values $\tr\left[\varrho_a\rho^{-1/2}\varrho_a\rho^{-1/2} \right]$ do not depend on $a$,
so it suffices to evaluate~\eqref{p1} for the simplest case, when $a=N$, then $V^{t_{n}}(n-1,n)=\mathbf{1}\ot P^+_{n-1,n}\equiv \mathbf{1}\ot P_+$. In this proof $P_+$ is an unnormalized projector onto a maximally entangled state between $(n-1)$-th and $n$-th subsystem - we will use notation $P_+$ and $V^{t_n}(n-1,n)$ interchangeably. 
Using Theorem~\ref{m1} we get:
	\be
	\label{above}
	\begin{split}
		F&=\frac{N}{d^{N+2}}\tr\left[\left(\mathbf{1}\ot P_+ \right)\eta^{-1/2}\left(\mathbf{1}\ot P_+ \right)\eta^{-1/2}\right]\\ &=\frac{N}{d^{N+2}}\sum_{\alpha,\alpha'}\sum_{\substack{\mu \in \alpha \\ \mu'\in \alpha'}}\frac{\zeta_{\mu,\mu'}(\alpha,\alpha')}{\sqrt{\gamma_{\mu}(\alpha)}\sqrt{\gamma_{\mu'}(\alpha')}} ,
		\end{split}
		\ee
where $\zeta_{\mu,\mu'}(\alpha,\alpha') = \tr\left[V^{t_n}(n-1,n)F_{\mu}(\alpha)V^{t_n}(n-1,n)F_{\mu'}(\alpha') \right]$.
Using  Fact~\ref{BF2} from Appendix~\ref{AppB} we rewrite the trace as follows:
\begin{align}
	\zeta_{\mu,\mu'}(\alpha,\alpha') &= \tr\left[V^{t_n}(n-1,n)P_{\mu}M_{\alpha}V^{t_n}(n-1,n)P_{\mu'}{M_{\alpha'}} \right]\\
	&=\tr\left[P_{\alpha'}V^{t_n}(n-1,n)P_{\mu}P_{\alpha}V^{t_n}(n-1,n)P_{\mu'} \right].	
	\end{align}
After further simplification (see Appendix~\ref{simpltheorem3} for details), we further reduce the expression inside of the trace to get:
\be\label{theorem3part1}
\zeta_{\mu,\mu'}(\alpha,\alpha') = \sum_{a,b=1}^{n-1}\sum_{i,j=1}^{d_{\mu}}\frac{d_{\mu}d_{\mu'}d^{\delta_{a,n-1}}d^{\delta_{b,n-1}}}{\left[ (n-1)!\right] ^2}\varphi_{ij}^{\mu}(a,n-1)\varphi_{kl}^{\mu'}(b,n-1)\tr\left[\left(P_{\alpha}F_{ij}^{\mu}\right) \left(  P_{\alpha'}F_{kl}^{\mu'}\right) \right],
\ee
{where $F^\mu_{ij}$ are defined in Fact~\ref{Fij} of Appendix~\ref{AppC}}.
The operators $F_{ij}^{\mu}, F_{kl}^{\mu'}$ can be expressed as direct sum of operators $E_{st}^{\beta}$ {as in Eqn.~\eqref{Fprop} of Appendix~\ref{AppC}}, where $\beta \vdash n-2$, so
	\be
	\label{trace}
	\tr\left[\left(P_{\alpha}F_{ij}^{\mu}\right) \left(  P_{\alpha'}F_{kl}^{\mu'}\right) \right]=\frac{\left[(n-2)! \right]^2}{d_{\alpha}^2}m_{\alpha}\delta_{\alpha \alpha'}\delta_{li}\delta_{jk}.
	\ee
	Substituting~\eqref{trace} into~\eqref{theorem3part1} and collecting the terms we obtain the following equation:
	\be
	\label{oop1}
	F=\frac{1}{Nd^{N+2}}\sum_{\alpha }\frac{m_{\alpha }}{d_{\alpha }^{2}}%
	\sum_{\substack{\mu \in \alpha \\ \mu'\in \alpha'}}\frac{d_{\mu }d_{\mu'}}{\sqrt{\gamma
			_{\mu }(\alpha )\gamma _{\mu'}(\alpha )}}g_{\mu,\mu'}(\alpha),
	\ee
	where the term $g_{\mu,\mu'}(\alpha)=\sum_{a,b=1}^{n-1}d^{%
		\delta _{a,n-1}}d^{\delta _{b,n-1}}\sum_{i_{\alpha },j_{\alpha
		}=1}^{d_{\alpha }}(\phi _{R}^{\mu })_{i_{\alpha }j_{\alpha }}^{\alpha \alpha
	}(a,n-1)(\phi _{R}^{\mu'})_{j_{\alpha }i_{\alpha }}^{\alpha \alpha
	}(b,n-1)$, evaluated in  Appendix~\ref{evalg}, yields the following tractable expression for the fidelity:
\be
	F=\frac{1}{Nd^{N+2}}\sum_{\alpha }\frac{m_{\alpha }}{d_{\alpha }}%
	\sum_{\substack{\mu \in \alpha \\ \mu'\in \alpha'}}d_{\mu }d_{\mu'}\sqrt{\gamma _{\mu}(\alpha) \gamma _{\mu'}(\alpha) }.
	\ee
	 
	We get the final result by substituting the expression for the eigenvalues $\gamma_{\mu}(\alpha), \gamma_{\mu'}(\alpha)$ from Eqn.~\eqref{m2eq1} of Theorem~\ref{m2}.
\end{proof}

\section*{Acknowledgments}
MS~is supported by the grant "Mobilno{\'s}{\'c} Plus IV", 1271/MOB/IV/2015/0 from the Polish Ministry of Science and Higher Education. MH and MM are supported by National Science Centre, Poland, grant OPUS 9.  2015/17/B/ST2/01945.

\section*{Competing financial interests}

The authors declare no competing financial interests.
\newpage

\appendix
\begin{center}
{\huge \bf Port-based teleportation in arbitrary dimension}\\
\vspace{4mm}
{\Large Supplementary Information}\\
\vspace{5mm}
{\large Micha{\l} Studzi{\'n}ski, Sergii Strelchuk, Marek Mozrzymas, Micha{\l} Horodecki}
\end{center}
\section{Auxilliary fact about operator $V^{t_n}(n-1,n)$}
\label{AppB}

\begin{fact}
	\label{BF2}
Let $M_{\alpha}$ be projector including multiplicities onto $\alpha$-th irrep of the algebra $\mathcal{A}_{n}^{\operatorname{t}_{n}}(d)$, let $P_{\alpha}, P_{\beta}$ where $\alpha,\beta \vdash n-2$ be a Young projectors, and let $V^{t_n}(n-1,n)$ be a  permutation operator acting between $(n-1)$-th and $n$-th subsystems partially transposed with respect to $n$-th subsystem, then
\be
M_{\alpha}V^{t_n}(n-1,n)=P_{\alpha}V^{t_n}(n-1,n),
\ee
\end{fact}

\begin{proof}
	The proof is based on the results presented in~\cite{Stu1}. Namely we know that operators $V^{t_n}(\sigma)$, where $\sigma \in S(n)$ can be decomposed in every irrep labelled by $\alpha$ in operator basis $\{v_{ij}^{ab}(\alpha)\}$, where $1\leq a,b \leq n-1$ and $1\leq i,j\leq d_{\alpha}$. In particular, when $\sigma=(n-1,n)$ we have:
	\be
	\label{12}
	V^{t_n}(n-1,n)=\sum_{\alpha \vdash n-2}\sum_{i,j=1}^{d_{\alpha }}\varphi_{ij}^{
	\alpha}(\operatorname{e})v_{ij}^{n-1,n-1}(\alpha),
	\ee
	where $\operatorname{e}$ denotes the identity element of $S(n-2)$. We see that  $\sum_{i,j=1}^{d_{\alpha }}\varphi_{ij}^{
		\alpha}(\operatorname{e})v_{ij}^{n-1,n-1}(\alpha)=M_{\alpha}V^{t_n}(n-1,n)$ is a restriction of $V^{t_n}(\sigma)$ to irrep labelled by $\alpha$, so rewriting equation~\eqref{12}
	\be
	\begin{split}
	M_{\alpha}V^{t_n}(n-1,n)=\sum_{i,j=1}^{d_{\alpha }}\varphi_{ij}^{\alpha}(\operatorname{e})v_{ij}^{n-1,n-1}(\alpha)=\sum_{i=1}^{d_{\alpha }}v_{ii}^{n-1,n-1}(\alpha),
		\end{split}
	\ee
	since $\varphi_{ij}^{\alpha}(\operatorname{e})=\delta_{ij}$. 
	Using equation~\eqref{ppp} from Section~\ref{Spectrum} we write:
	\be
	\sum_{i=1}^{d_{\alpha }}v_{ii}^{n-1,n-1}(\alpha)=\sum_{i=1}^{d_{\alpha }}E_{ii}^{\alpha}V^{t_n}(n-1,n)=P_{\alpha}V^{t_n}(n-1,n),
	\ee
	since $P_{\alpha}=\sum_{i=1}^{d_{\alpha }}E_{ii}^{\alpha}$. This finishes the proof.
\end{proof}
	
\section{Partially reduced irreducible representations (PRIR)}
\label{AppA}
In this section we derive some properties of the PRIR. This concept plays a crucial role in the simplification of the representation of the algebra $\mathcal{A}_n^{t_n}(d)$.

Let us consider an arbitrary unitary irrep $\phi^{\mu }$ of $S(n)$. It
can be always unitarily transformed to PRIR $\phi_{R}^{\mu }$ such that 
\be
\forall \pi \in S(n-1)\quad \phi _{R}^{\mu }(\pi )=\bigoplus _{\alpha \in \mu
}\varphi ^{\alpha }(\pi), 
\ee
where $\varphi ^{\alpha }$ are irreps of $S(n-1)$. By $\alpha \in \mu$ we denote Young diagrams $\alpha$ which can be obtained from $\mu$ by removing one box in the proper way. We see that the restriction of the irrep  $\phi^{\mu }$ of $S(n)$ to the subgroup $S(n-1)$ has a block-diagonal form of completely reduced representation, which in matrix notation takes the form%
\be
\label{prir1}
\forall \pi \in S(n-1)\quad \phi _{R}^{\mu }(\pi)=\left( \delta ^{\alpha \beta
}\varphi _{i_{\alpha }j_{\alpha }}^{\alpha }\right) .
\ee
The block structure of this reduced representation allows us to introduce such a block indexation for PRIR $\phi _{R}^{\mu }$ of $S(n)$, which gives 
\be
\forall \sigma \in S(n)\quad \phi _{R}^{\mu }(\sigma )=\left( \phi _{i_{\alpha
	}j_{\beta }}^{\alpha \beta }(\sigma )\right) , 
\ee
where the matrices on the diagonal $(\phi _{R}^{\mu })^{\alpha \alpha
}(\sigma )=\left( \phi _{i_{\alpha }j_{\alpha }}^{\alpha \alpha }(\sigma )\right) $ are
of dimension of corresponding irrep $\varphi ^{\alpha }$ of $S(n-1)$. The
off diagonal blocks need not to be square. The matrices $(\phi _{R}^{\mu
})^{\alpha \alpha }(\sigma )=\left( \phi _{i_{\alpha }j_{\alpha }}^{\alpha \alpha
}(\sigma )\right) $ on the diagonal of the matrix $\phi _{R}^{\mu }(\sigma )$ have
the following important properties:

\begin{proposition}
	\label{MP1}
	Let $(\phi _{R}^{\mu })^{\alpha \alpha }(\sigma )=\left( \phi _{i_{\alpha
		}j_{\alpha }}^{\alpha \alpha }(\sigma )\right) $ be the matrices on the diagonal of
	the PRIR matrix $\phi _{R}^{\mu }(\sigma )$, then%
	\be
	\forall \alpha \in \mu \quad \varphi ^{\alpha }(\pi )\left( \phi _{R}^{\mu
	})^{\alpha \alpha }(an)\right) \varphi ^{\alpha }(\pi^{-1})=\left( \phi _{R}^{\mu
	})^{\alpha \alpha }(\pi(a)n)\right) , 
	\ee
	and from this it follows%
	\be
	\forall \alpha \in \mu \quad \forall \pi \in S(n-1)\quad \forall
	a=1,\ldots,n-1\qquad \tr\left[ (\phi _{R}^{\mu })^{\alpha \alpha }(an)\right] =\tr\left[ (\phi _{R}^{\mu
	})^{\alpha \alpha }(\pi (a)n)\right] , 
	\ee
	so the trace is constant on the transpositions which naturally indexed the coset $S(n)/S(n-1).$
\end{proposition}

\begin{proof}
	From the composition rule in $S(n)$ we have 
	\be
	\forall \pi \in S(n-1)\quad \forall a=1,\ldots,n-1\qquad \quad \pi \circ (an)\circ \pi
	^{-1}=(\pi (a)n), 
	\ee
	which implies 
	\be
	\label{a7}
	\forall \pi \in S(n-1)\quad \forall a=1,\ldots,n-1\qquad \phi _{R}^{\mu }(\pi
	)\phi _{R}^{\mu }(an)\phi _{R}^{\mu }(\pi ^{-1})=\phi _{R}^{\mu }(\pi
	(a)n), 
	\ee
	where the matrices $\phi _{R}^{\mu }(\pi )$, $\phi _{R}^{\mu
	}(\pi ^{-1})$ are block-diagonal (see expression~\eqref{prir1}). From
	multiplication rule of block diagonal matrices we get that the 
	equation~\eqref{a7} in the irrep $\phi _{R}^{\mu }$ yields the following equations
	for its diagonal blocks determined by the irrep's $\varphi ^{\alpha }$
	of $S(n-1)$%
	\be
	\forall \alpha \in \mu \quad \varphi ^{\alpha }(\pi )\left( \phi _{R}^{\mu
	})^{\alpha \alpha }(an)\right) \varphi ^{\alpha }(\pi^{-1})=\left( \phi _{R}^{\mu
	})^{\alpha \alpha }(\pi(a)n)\right) . 
	\ee
	Taking the trace, in $\M(d_{\alpha},\mathbb{C})$, on this equation we get the second statement of the Proposition.
\end{proof}

Further we have the following sum rules

\begin{proposition}
	\label{Prir1}
	The PRIR $\phi _{R}^{\mu }$ of $S(n)$ satisfies the following sum rules%
	\be
	\label{Prir1eq}
	\sum_{a=1}^{n-1}(\phi _{R}^{\mu })(an)=\frac{n(n-1)}{2}\frac{\chi ^{\mu
		}(12)}{d_{\mu }}\mathbf{1}_{\phi ^{\mu }}-\bigoplus _{\alpha \in \mu }\frac{%
		(n-1)(n-2)}{2}\frac{\chi ^{\alpha }(12)}{d_{\alpha }}\mathbf{1}_{\varphi
		^{\alpha }}, 
	\ee
	which implies that for the diagonal blocks we have 
	\be
	\forall \alpha \in \mu \qquad \sum_{a=1}^{n-1}(\phi _{R}^{\mu })^{\alpha
		\alpha }(an)=\left[ \frac{n(n-1)}{2}\frac{\chi ^{\mu }(12)}{d_{\mu }}%
	-\frac{(n-1)(n-2)}{2}\frac{\chi ^{\alpha }(12)}{d_{\alpha }}%
	\right] \mathbf{1}_{\varphi ^{\alpha }}. 
	\ee
\end{proposition}

\begin{proof}
	The starting point is the classical equation
	\be
	\label{P11p}
	\sum_{(ab)\in S(n)}(\phi _{R}^{\mu })(ab)=\frac{n(n-1)}{2}\frac{\chi ^{\mu
		}(12)}{d_{\mu }}\mathbf{1}_{\phi ^{\mu }}, 
	\ee
	which holds for any irrep of $S(n)$. We rewrite $LHS$ of equation~\eqref{P11p} separating the terms
	in $S(n-1)$ 
	\be
	\label{terms}
	\begin{split}
		\sum_{(ab)\in S(n)}(\phi _{R}^{\mu })(ab)&=\sum_{a=1}^{n-1}(\phi _{R}^{\mu
		})(an)+\sum_{(cd)\in S(n-1)}(\phi _{R}^{\mu })(cd)\\ 
		&=\sum_{a=1}^{n-1}(\phi _{R}^{\mu })(an)+\bigoplus _{\alpha \in \mu
		} \ \sum_{(cd)\in S(n-1)}\varphi ^{\alpha }(cd). 
	\end{split}
	\ee
	Now we use one more equation~\eqref{P11p}  to each irrep $%
	\varphi ^{\alpha }$ in the direct sum of $RHS$ in equation~\eqref{terms} we get
	\be
	\sum_{(ab)\in S(n)}(\phi _{R}^{\mu })(ab)=\sum_{a=1}^{n-1}(\phi _{R}^{\mu
	})(an)+\bigoplus _{\alpha \in \mu }\frac{(n-1)(n-2)}{2}\frac{\chi ^{\alpha }(12)%
	}{d_{\alpha }}\mathbf{1}_{\varphi ^{\alpha }}. 
	\ee
	This equation together with expression~\eqref{P11p} gives the first statement of the
	proposition.
\end{proof}

\begin{remark}
	Equation~\eqref{Prir1eq} in Proposition~\ref{Prir1} may be written in a  more
	explicit form as follows:
	\be
	\forall \alpha \in \mu \qquad \sum_{a=1}^{n-1}(\phi _{R}^{\mu })_{i_{\alpha
		}j_{\alpha }}^{\alpha \alpha }(an)=\left[ \frac{n(n-1)}{2}\frac{\chi ^{\mu }(12)%
	}{d_{\mu }}-\frac{(n-1)(n-2)}{2}\frac{\chi ^{\alpha }(12)}{d_{\alpha }}\right] \delta _{i_{\alpha }j_{\alpha }},
	\ee
	where $i_{\alpha},j_{\alpha}=1,\ldots,d_{\alpha}$.
\end{remark}

We have one more summation rule, which plays a role of the standard orthogonality
relation for irreps. Namely we have the following:

\begin{proposition}
	\label{AP13}
	The PRIR $\phi _{R}^{\mu }$ of $S(n)$ satisfies the following bilinear sum rule%
	\be
	\forall \alpha ,\beta ,\gamma \in \mu \qquad \sum_{a=1}^{n}\sum_{k_{\beta
		}=1}^{d_{\beta}}(\phi _{R}^{\mu })_{i_{\alpha }k_{\beta }}^{\alpha \beta
	}(an)(\phi _{R}^{\mu })_{k_{\beta }j_{\gamma }}^{\beta \gamma }(an)=n\frac{%
		d_{\beta }}{d_{\mu }}\delta ^{\alpha \gamma }\delta
	_{i_{\alpha }j_{\gamma }}, 
	\ee
	where $\alpha, \beta, \gamma $ are irreps of $S(n-1)$ contained
	in the irrep $\mu $ of $S(n)$.
\end{proposition}

\begin{proof}
	The proof is based on the standard orthogonality relations for irreps,
	which in PRIR notation take the following form%
	\be
	\label{Prop13eq1}
	\forall \alpha ,\beta ,\gamma \in \mu \qquad \sum_{\sigma \in S(n)}(\phi
	_{R}^{\mu })_{i_{\alpha }k_{\beta }}^{\alpha \beta }(\sigma ^{-1})(\phi
	_{R}^{\mu })_{k_{\beta }j_{\gamma }}^{\beta \gamma }(\sigma )=\frac{n!}{d_{\mu }}\delta ^{\alpha \gamma }\delta _{i_{\alpha }j_{\gamma }},
	\ee
	for any irreps $\alpha ,\beta ,\gamma $ of the group $S(n-1)$ which are
	contained in the irrep $\mu $ of $S(n)$. On the other hand we may rewrite
	the $LHS$ of the above equation as follows 
	\be
	LHS=\sum_{a=1}^{n} \ \sum_{\pi \in S(n-1)} \ \sum_{\xi ,\theta \in \mu
	} \ \sum_{p_{\xi },q_{\theta }}(\phi _{R}^{\mu })_{i_{\alpha p_{\xi }}}^{\alpha
		\xi }(an)(\phi _{R}^{\mu })_{p_{\xi }k_{\beta }}^{\xi \beta }(\pi
	^{-1})(\phi _{R}^{\mu })_{k_{\beta }q_{\theta }}^{\beta \theta }(\pi )(\phi
	_{R}^{\mu })_{q_{\theta }j_{\gamma }}^{\theta \gamma }(an).
	\ee
	Taking into account equation~\eqref{prir1}  we obtain
	\be
	LHS=\sum_{a=1}^{n} \ \sum_{\pi \in S(n-1)} \ \sum_{p_{\beta },q_{\beta }}(\phi
	_{R}^{\mu })_{i_{\alpha p_{\beta }}}^{\alpha \beta }(an)\varphi ^{\beta
	}{}_{p_{\beta }k_{\beta }}(\pi ^{-1})\varphi ^{\beta }_{k_{\beta
		}q_{\beta }}(\pi)(\phi _{R}^{\mu })_{q_{\beta }j_{\gamma }}^{\beta \gamma
	}(an),
	\ee
	next applying the orthogonality relations for irreps $\varphi ^{\beta}$ of $S(n-1)$ we get 
	\be
	LHS=\frac{(n-1)!}{d_{\beta }}\sum_{a=1}^{n}\sum_{p_{%
			\beta }}^{d_{\beta }}(\phi _{R}^{\mu })_{i_{\alpha
			p_{\beta }}}^{\alpha \beta }(an)(\phi _{R}^{\mu })_{p_{\beta }j_{\gamma
	}}^{\beta \gamma }(an).
	\ee
	Now comparing this with the $RHS$ of the equation~\eqref{Prop13eq1}, we obtain
	the statement of the proposition.
\end{proof}
As a corollary from Propositions~\ref{MP1} and Proposition~\ref{Prir1} we get

\begin{corollary}
	\label{Cc}
	\be
	\forall \alpha \in \mu \quad \forall a=1,\ldots,n-1\qquad \tr\left[ (\phi _{R}^{\mu
	})^{\alpha \alpha }(an)\right] =\frac{n}{2}\frac{d_{\alpha }}{d_{\beta }}\chi ^{\mu }(12)-\frac{n-2}{2}\chi ^{\alpha }(12). 
	\ee
\end{corollary}

\section{Auxiliary facts concerning Young projectors}
\label{AppC}
Let us define the following set of permutations
\be
\label{dec1}
\Sigma_a=\left\lbrace \sigma \in S(n-1) \ : \ \sigma(a)=n-1\right\rbrace,\quad \text{then we have} \quad S(n-1)=\bigcup_{a=1}^{n-1}\Sigma_a.
\ee
Now we see that for every $\sigma \in \Sigma_a$ permutation $\sigma \circ (a,n-1)$ belongs to $S(n-2)$, since $(\sigma \circ (a,n-1))(n-1)=n-1$. Such property allows us to rewrite Young projectors $P_{\mu}$, where $\mu \vdash n-1$ in a more convenient form, namely we have the following:
\begin{fact}
	\label{Fij}
	Young projector $P_{\mu}$, where $\mu \vdash n-1$ can be written as
	\be
	P_\mu=\frac{d_{\mu}}{(n-1)!}\sum_{a=1}^{n-1}\sum_{i,j=1}^{d_{\mu}}\varphi_{ij}^{\mu}(a,n-1)V(a,n-1)F_{ij}^{\mu},
	\ee
	where
	\be
	\label{mal}
	F_{ij}^{\mu}=\sum_{\pi \in S(n-2)} \varphi_{ji}^{\mu}\left(\pi^{-1} \right)V(\pi). 
	\ee
	By $\varphi_{ij}^{\mu}(a,n-1), \varphi_{ji}^{\mu}\left(\pi^{-1} \right) $ we denote matrix elements of irreducible representations labelled by partition $\mu$ for the permutations $(a,n-1), \pi^{-1}$ respectively. Note that in the equation~\eqref{mal} we compute matrix elements of irreducible representations for partition $\mu \vdash n-1$, but over subgroup $S(n-2)\subset S(n-1)$.
\end{fact}

\begin{proof}
	Proof is based on straightforward calculations and observations summarized in the formula~\eqref{dec1}. We have the following chain of equalities:
	\be
	\begin{split}
		P_{\mu}&=\frac{d_{\mu}}{(n-1)!}\sum_{\sigma \in S(n-1)}\chi^{\mu}\left(\sigma^{-1} \right)V(\sigma)=\frac{d_{\mu}}{(n-1)!}\sum_{a=1}^{n-1}\sum_{\pi\in S(n-2)}\chi^{\mu}\left((a,n-1)\circ \pi^{-1} \right) V\left((a,n-1)\circ \pi \right)\\
		&=d_{\mu}\sum_{a=1}^{n-1}\sum_{i,j=1}^{d_{\mu}}\varphi_{ij}^{\mu}(a,n-1)V(a,n-1)\left(\frac{1}{(n-1)!}\sum_{\pi\in S(n-2)}\varphi_{ji}^{\mu}\left(\pi^{-1}\right)V(\pi) \right)\\
		&=d_{\mu}\sum_{a=1}^{n-1}\sum_{i,j=1}^{d_{\mu}}\varphi_{ij}^{\mu}(a,n-1)V(a,n-1)F_{ij}^{\mu}.
	\end{split}
	\ee
\end{proof}

 Every irreducible block labelled by $\mu \vdash n-1$ can be decomposed as a direct sum of smaller irreducible blocks labelled by partitions $\beta \vdash n-2$. Every such partition $\beta$ is obtained by removing a single box from $\mu$ in the proper way. This togehter with the notion of PRIRs defined in Appendix~\ref{AppA} allows us to decompose every $F_{ij}^{\mu}$ from Fact~\ref{Fij} as
\be
\label{fij}
F_{ij}^{\mu}=\bigoplus_{\beta=\mu-\Box} \ \sum_{\pi \in S(n-2)}(\varphi^{\mu}_R)^{\beta \beta}_{i_{\beta}j_{\beta}}\left(\pi^{-1}\right)V(\pi).
\ee
 Moreover, every operator $F_{ij}^{\mu}$ can be expressed in therms of the projectors $E_{ij}^{\beta}$  as
\be
\label{Fprop}
F_{ij}^{\mu}=\bigoplus_{\beta=\mu-\Box}\frac{(n-2)!}{d_{\beta}}E^{\beta}_{i_{\beta}j_{\beta}}.
\ee

\begin{fact}
	\label{fact11}
	Suppose that we are given an irreducible representation labelled by $\mu \vdash n-1$, then for every swap operator $V(k,n-1)$ between $k^{\text{th}}$ and $(n-1)^{\text{th}}$ subsystem, and Young projector $P_{\mu}$ we have
	\be
	\sum_{k=1}^{n-1} V(k,n-1)P_{\mu}V(k,n-1)=(n-1)P_{\mu}.
	\ee
\end{fact}

\begin{proof}
	We know that every Young projector associated with irreducible representation $\mu$ can be written as
	\be
	P_{\mu}=\frac{d_{\mu}}{(n-1)!}\sum_{\sigma \in S(n-1)}\chi^{\mu}\left( \sigma^{-1}\right)V(\sigma),
	\ee
	where $\chi^{\mu}\left( \sigma^{-1}\right)$ is the character of irreducible representation $\mu$ calculated on the element $\sigma^{-1} \in S(n-1)$, and $V(\sigma)$ is the permutation operator acting on $\left(\mathbb{C}^d \right)^{\ot (n-1)}$. Since operator 
	$P_{\mu}$ belongs to the centre of the algebra $\mathbb{C}\left[S(n-1) \right] $ it commutes with all elements $V(\sigma)\in \mathbb{C}\left[S(n-1) \right] $, where $\sigma \in S(n-1)$ in particular with $V(k,n-1)\in \mathbb{C}\left[S(n-1) \right] $ for $k=1,\ldots,n-1$. This finishes the proof.
\end{proof}

\begin{fact}
	\label{obkl}
	Let us denote by $P_+$  projector onto unnormalized maximally entangled state $|\psi^+\>=\sum_i |ii\>$ between $(n-1)^{\text{th}}$ and $n^{\text{th}}$ subsystem, then:
	\be
	\left( \mathbf{1}\ot P_+\right)  V(k,n-1)\left( \mathbf{1}\ot P_+\right)  =\begin{cases}
		d\left(\mathbf{1}\ot P_+ \right)\quad \text{if} \quad k=n-1,\\
		\mathbf{1}\ot P_+  \quad \text{if} \quad k=1,\ldots, n-2.
	\end{cases}
	\label{cases}
	\ee
	In the above, by $V(k,n-1)$ we denote swap operator between $k$-th and $(n-1)$-th subsystem respectively, and by $d$ dimension of the local Hilbert space.
\end{fact}

\begin{proof}
	For $k=n-1$ we have simply $\left(\mathbf{1}\ot P_+ \right)^2=d\left(\mathbf{1}\ot P_+ \right)$, since $P_+$ is unnormalized. Now we have to prove the second case from the formula~\eqref{cases}:
	\be
	\begin{split}
		&\left( \mathbf{1}\ot P_+\right)  V(k,n-1)\left( \mathbf{1}\ot P_+\right)=\left(\sum_{j_n,j_{n-1}=1}^d \mathbf{1}_1\ot \cdots \ot \mathbf{1}_k\ot \cdots \ot |j_n\>\<j_{n-1}|\ot |j_n\>\<j_{n-1}|\right)\\
		&\times \left(\sum_{i_k,i_{n-1}=1}^d \mathbf{1}_1\ot \cdots \ot \mathbf{1}_{k-1}\ot |i_{n-1}\>\<i_k|\ot \cdots \ot |i_k\>\<i_{n-1}|\ot \mathbf{1}_n\right)  \\
		&\times \left(\sum_{l_{n-1},l_n=1}^d \mathbf{1}_1\ot \cdots \ot \mathbf{1}_k \ot \cdots \ot |l_n\>\<l_{n-1}|\ot |l_n\>\<l_{n-1}|\right)
	\end{split}
	\ee
	\be
	\begin{split}
		&=\sum_{\substack{j_n,j_{n-1}=1\\i_k,i_{n-1}=1\\l_{n-1},l_n=1}}^d \mathbf{1}_1\ot \cdots \ot \mathbf{1}_{k-1}\ot |i_{n-1}\>\<i_k|\ot \cdots \ot |j_n\>\<j_{n-1}|i_k\>\<i_{n-1}|l_n\>\<l_{n-1}|\ot |j_n\>\<j_{n-1}|l_n\>\<l_{n-1}|\\
		&=\sum_{i_k,j_n,l_{n-1}=1}^d \mathbf{1}_1\ot \cdots \ot \mathbf{1}_{k-1}\ot |i_k\>\<i_k|\ot \cdots \ot |j_n\>\<l_{n-1}|\ot |j_n\>\<l_{n-1}|=\mathbf{1}\ot P_+.
	\end{split}
	\ee
\end{proof}

\begin{fact}
	\label{F3app}
	For an arbitrary element $X$ of algebra $\mathcal{A}_{n}^{\operatorname{t}_{n}}(d)$, $\tr_n X \in \mathbb{C}[S(n-1)]$.
\end{fact}

\begin{proof}
	From~\cite{Moz1,Stu1} we know that algebra $\mathcal{A}_{n}^{\operatorname{t}_{n}}(d)$ is spanned by the partially transposed permutation operators $V^{t_n}(\sigma)$, where $\sigma \in S(n)$. Let us take an arbitrary operator $A$ defined on $n-1$ subsystems, then
	we can write
	\be
	\tr\left[V^{t_n}(\sigma) A\ot \mathbf{1}_n \right] =\tr\left[V(\sigma)A\ot \mathbf{1}_n^{t_n} \right],
	\ee
	where $\mathbf{1}_n$ is the identity operator on last system, and $t_n$ denotes standard transposition operation on last $n$-th system. We can now express the trace $\tr_n\left[V^{t_n}(\sigma) \right]=\tr\left[V(\sigma) \right]$, but $\tr\left[V(\sigma) \right]$ for $\sigma \in S(n)$ belongs to $\mathbb{C}[S(n-1)]$, so we have  $\tr_n\left[V^{t_n}(\sigma) \right]\in \mathbb{C}[S(n-1)]$.
\end{proof}
\section{Proof of the auxiliary statements in Theorem 3}\label{simpltheorem3}
\subsection{Proof of Eqn.~\eqref{theorem3part1}}
   Using the expression for $P_{\mu}$ from Fact~\ref{Fij} and applying Fact~\ref{obkl} we get:
	\be
	\begin{split}
		\zeta_{\mu,\mu'}(\alpha,\alpha') &=\sum_{a=1}^{n-1}\sum_{ij=1}^{d_{\mu}}\frac{d_{\mu}}{(n-1)!}\varphi_{ij}^{\mu}(a,n-1)\tr\left[\left(P_{\alpha}F_{ij}^{\mu}\ot \mathbf{1} \right)\left(\mathbf{1}\ot P_+ \right)V(a,n-1)\left(\mathbf{1}\ot P_+ \right)\left(P_{\alpha'}\ot \mathbf{1} \right)  P_{\mu'}\right]\\
		&=\sum_{a=1}^{n-1}\sum_{i,j=1}^{d_{\mu}}\frac{d_{\mu}d^{\delta_{a,n-1}}}{(n-1)!}\varphi_{ij}^{\mu}(a,n-1)\tr\left[\left(P_{\alpha}F_{ij}^{\mu}P_{\alpha'}\ot P_+ \right) P_{\mu'}\right]\\
		&=\sum_{a=1}^{n-1}\sum_{i,j=1}^{d_{\mu}}\frac{d_{\mu}d^{\delta_{a,n-1}}}{ (n-1)!}\varphi_{ij}^{\mu}(a,n-1)\tr\left[\left(P_{\alpha}F_{ij}^{\mu}P_{\alpha'}\ot \mathbf{1} \right) P_{\mu'}\right].
	\end{split}
	\ee
	In the last equality we have used the fact that $\tr_{n}P^+=\mathbf{1}$, where identity acts on $(n-1)$-th subsystem. Applying Fact~\ref{Fij} to operator $P_{\mu'}$ and calculating a partial trace over $(n-1)$-th subsystem we get:
	\be
	\begin{split}
		\zeta_{\mu,\mu'}(\alpha,\alpha') =\sum_{a,b=1}^{n-1}\sum_{i,j=1}^{d_{\mu}}\frac{d_{\mu}d_{\mu'}d^{\delta_{a,n-1}}d^{\delta_{b,n-1}}}{\left[ (n-1)!\right] ^2}\varphi_{ij}^{\mu}(a,n-1)\varphi_{kl}^{\mu'}(b,n-1)\tr\left[\left(P_{\alpha}F_{ij}^{\mu}\right) \left(  P_{\alpha'}F_{kl}^{\mu'}\right) \right].
	\end{split}
	\ee
\subsection{Evaluating $g_{\mu,\mu'}(\alpha)$}\label{evalg}
\be g_{\mu,\mu'}(\alpha)
	= \sum_{a,b=1}^{n-1}d^{%
		\delta _{a,n-1}}d^{\delta _{b,n-1}}\sum_{i_{\alpha },j_{\alpha
		}=1}^{d_{\alpha }}(\phi _{R}^{\mu })_{i_{\alpha }j_{\alpha }}^{\alpha \alpha
	}(a,n-1)(\phi _{R}^{\mu'})_{j_{\alpha }i_{\alpha }}^{\alpha \alpha
	}(b,n-1),
\ee
 with all irreps in the PRIR's form defined in Appendix~\ref{AppA}. Remarkably, this form allows us to directly evaluate these sums. First we partition the sums as follows:
	\be
	\begin{split}
		&\sum_{a,b=1}^{n-1}d^{\delta _{a,n-1}}d^{\delta _{b,n-1}}\sum_{i_{\alpha
			},j_{\alpha }=1}^{d_{\alpha }}(\phi _{R}^{\mu })_{i_{\alpha }j_{\alpha
		}}^{\alpha \alpha }(an-1)(\phi _{R}^{\mu'})_{j_{\alpha }i_{\alpha
		}}^{\alpha \alpha }(bn-1)\\
		&=\sum_{i_{\alpha },j_{\alpha }=1}^{d_{\alpha }}\left\lbrace \left[ \sum_{a=1}^{n-1}d^{\delta
			_{a,n-1}}(\phi _{R}^{\mu })_{i_{\alpha }j_{\alpha }}^{\alpha \alpha
		}(a,n-1)\right] \left[ \sum_{b=1}^{n-1}d^{\delta _{b,n-1}}(\phi _{R}^{\mu'})_{j_{\alpha }i_{\alpha }}^{\alpha \alpha }(b,n-1)\right] \right\rbrace \\
		&=\sum_{i_{\alpha },j_{\alpha }=1}^{d_{\alpha }}\left\lbrace \left[ \sum_{a=1}^{n-2}(\phi
		_{R}^{\mu })_{i_{\alpha }j_{\alpha }}^{\alpha \alpha }(a,n-1)+d\delta
		_{i_{\alpha }j_{\alpha }}\right] \left[ \sum_{b=1}^{n-2}(\phi _{R}^{\mu'})_{j_{\alpha }i_{\alpha }}^{\alpha \alpha }(b,n-1)+d\delta _{j_{\alpha
			}i_{\alpha }}\right]\right\rbrace .
	\end{split}
	\ee
	Now we use the second statement of Proposition~\ref{Prir1} for $S(n-1)$, for the sums over $a$ and $b$, which yields 
	\be
	\begin{split}
		&\sum_{i_{\alpha },j_{\alpha }=1}^{d_{\alpha }}\left\lbrace \left[ \left( \frac{(n-1)(n-2)}{2}\frac{%
			\chi ^{\mu }(12)}{d_{\mu}}-\frac{(n-2)(n-3)}{2}\frac{\chi
			^{\alpha }(12)}{d_{\alpha}}\right) \delta _{i_{\alpha }j_{\alpha
		}}+d\delta _{i_{\alpha }j_{\alpha }}\right] \times \right. \\
		&\left. \times\left[  \left( \frac{(n-1)(n-2)}{2}\frac{\chi ^{\mu' }(12)}{d_{\mu'}}-\frac{(n-2)(n-3)}{2}\frac{\chi ^{\alpha }(12)}{%
			d_{\alpha}}\right) \delta _{j_{\alpha }i_{\alpha }}+d\delta
		_{j_{\alpha }i_{\alpha }}\right]\right\rbrace,
	\end{split}
	\ee
	and after simple reordering we get 
	\be
	\label{above1}
	\begin{split}
		&\left[  \left( \frac{(n-1)(n-2)}{2}\frac{\chi ^{\mu }(12)}{d_{\mu }}-
		\frac{(n-2)(n-3)}{2}\frac{\chi ^{\alpha }(12)}{d_{\alpha }}
		\right) +d\right] \times \\
		&\left[  \left( \frac{(n-1)(n-2)}{2}\frac{\chi ^{\mu'}(12)}{d_{\mu' }}-\frac{(n-2)(n-3)}{2}\frac{\chi ^{\alpha }(12)}{d_{\alpha }}\right) +d\right] \sum_{i_{\alpha },j_{\alpha }=1}^{d_{\alpha }}\delta
		_{i_{\alpha }j_{\alpha }}\delta _{j_{\alpha }i_{\alpha }}.
	\end{split}
	\ee
	In equation~\eqref{above1} we recognize inside the square brackets the expression for eigenvalues $%
	\gamma _{\mu }(\alpha )$ and $\gamma _{\mu ^{\prime }}(\alpha )$, this yields
	\be
	\label{above2}
	\sum_{a,b=1}^{n-1}d^{\delta _{a,n-1}}d^{\delta _{b,n-1}}\sum_{i_{\alpha
		},j_{\alpha }=1}^{d_{\alpha }}(\phi _{R}^{\mu })_{i_{\alpha }j_{\alpha
	}}^{\alpha \alpha }(a,n-1)(\phi _{R}^{\mu'})_{j_{\alpha }i_{\alpha
	}}^{\alpha \alpha }(b,n-1)=\gamma _{\mu }(\alpha )\gamma _{\mu'}(\alpha )d_{\alpha }.
	\ee
	Substituting equation~\eqref{above2} into equation~\eqref{oop1} we reduce expression for the fidelity $F$ to:
	\be
	F=\frac{1}{Nd^{N+2}}\sum_{\alpha }\frac{m_{\alpha }}{d_{\alpha }}%
	\sum_{\substack{\mu \in \alpha \\ \mu'\in \alpha'}}d_{\mu }d_{\mu'}\sqrt{\gamma _{\mu}(\alpha) \gamma _{\mu'}(\alpha) }.
	\ee

\section{Auxiliary facts concerning multiplicities of symmetric group}
\label{app:multSn}

Let us consider the standard swap representation 
\be
V_{n}^{d}:S(n)\rightarrow \Hom \left[ (\mathbb{C}^{d})^{\otimes n}\right] .
\ee
It is well-known that it reduces to the form 
\be
V_{n}^{d}\cong \bigoplus _{\nu :h(\nu )\leq d}m_{\nu }\psi ^{\nu },
\ee
where $\psi ^{\nu }$ are irreps of $S(n)$, $m_{\nu }$ their
multiplicities, and by $h(\nu)$ we denote height of Young diagram $\nu$. From theory of group characters we get 

\begin{proposition}
	\be
	\sum_{\nu :h(\nu )\leq d}m_{\nu }^{2}=\frac{1}{n!}\sum_{\sigma \in
		S(n)}d^{2l_{n}(\sigma )},
	\ee
	where $l_{n}(\sigma )$ is the number of cycles in the permutation $\sigma $
	as a permutation of $S(n)$.
\end{proposition}

Now using the following 

\begin{lemma}
	Let 
	\be
	\sigma \in S(n)\wedge \sigma =(an)\rho :\rho \in S(n-1),\quad a=1,\ldots,n,
	\ee
	then 
	\be
	l_{n}(\sigma )=%
	\begin{cases}
	l_{n-1}(\rho ) \ : \ a\neq n, \\ 
	l_{n-1}(\rho )+1 \ : \ a=n
	\end{cases}
	\ee
	so for the non trivial cosets $S(n)/S(n-1)$ the number of cycles is the same.
\end{lemma}
We can now express the ratio of the multiplicities in closed form:
\begin{proposition}
	\label{P23}
	Let 
	\be
	V_{n}^{d}\cong \bigoplus _{\nu :h(\nu )\leq d}m_{\nu }\psi ^{\nu },\qquad
	V_{n-1}^{d}\cong \bigoplus _{\alpha :h(\alpha )\leq d}m_{\alpha }\varphi
	^{\alpha },
	\ee
	then 
	\be
	\label{pr23}
	\frac{\sum_{\nu :h(\nu )\leq d}m_{\nu }^{2}}{\sum_{\alpha :h(\alpha )\leq
			d}m_{\alpha }^{2}}=\frac{d^{2}+n-1}{n}.
	\ee

\end{proposition}

\begin{lemma}
	\label{L26}
	For any Young diagram $\mu: h(\mu)\leq d$  we have 
	\be
	\frac{1}{m_{\mu }}\sum_{\alpha \in \mu}\gamma _{\mu}(\alpha)m_{\alpha }=n-1, 
	\ee
	where $\alpha \in \mu$ denotes Young diagrams of $n-2$ which are obtained from Young diagrams of $n-1$ by removing one box in a proper way, and numbers $\gamma_{\mu}(\alpha)$ are given in Proposition~\ref{m2}.
\end{lemma}

\begin{proof}
	Using the explicit form of the numbers $\gamma_{\mu}(\alpha)$ we have 
	\be
	\frac{1}{m_{\mu }}\sum_{\alpha \in \mu}\gamma _{\mu}(\alpha) m_{\alpha }=\frac{n-1}{m_{\mu }}\sum_{\alpha\in \mu}\frac{m_{\mu
		}d_{\alpha }}{m_{\alpha }d_{\mu }}m_{\alpha }=(n-1)\frac{1}{d_{\mu }}%
	\sum_{\alpha \in \mu}d_{\alpha }. 
	\ee
	The dimensions $d_{\alpha }$ corresponding to $\alpha \in \mu$ are precisely the
	dimensions of irreps of $S(n-2)$ appearing in  the restriction of irreps of $S(n-1)$ to $S(n-2)$, so
	\be
	\sum_{\alpha \in \mu}d_{\alpha }=d_{\mu }.
	\ee
	This finishes the proof.
\end{proof}

\section{Description of the operators $E_{ij}^{\alpha}$}
\label{A2}
In this section we will briefly recall some properties of the algebra
generated by a given complex finite dimensional representation of the finite
group $G.$ The content of this section can be found in standard textbooks on representation theory of finite groups and algebras, for example in \cite{Curtis}.

Any complex finite-dimensional representation $D:G\rightarrow \Hom(V)$ of the
finite group $G,$ where $V$ is a complex linear space, generates an algebra $%
A_{V}[G]$ $\subset $ $\Hom(V)$ which is isomorphic to the group algebra $\mathbb{C}\lbrack G]$ if the representation $D$ is faithful:
\be
A_{V}[G]=\operatorname{span}_{\mathbb{C}}\{D(g),\quad g\in G\}.
\ee
If the operators $D(g)$ are linearly independent, then they form a basis of
the algebra $A_{V}[G]$ and $\dim A_{V}[G]=\left\vert G\right\vert $. It is
also possible, using matrix irreducible representations, to construct a new
basis which has remarkable properties, very useful in applications of
representation theory. Below we describe this construction.

\begin{notation}
	\label{not1}
	Let $G$ be a finite group of order $\left\vert G\right\vert$ which has $r$
	classes of conjugated elements. Then $G$ has exactly $r$ inequivalent,
	irreducible representations, in particular $G$ has exactly $r$ inequivalent,
	irreducible matrix representations. Let
	\be
	D^{\alpha }:G\rightarrow \Hom(V^{\alpha }),\qquad \alpha =1,2,\ldots,r,\qquad
	\dim V^{\alpha }=d_{\alpha }
	\ee
	be all inequivalent, irreducible representations of $G$ and let us assume that these representations are all unitary (always possible) i.e.
	\be
	D^{\alpha }(g)=(D_{ij}^{\alpha }(g)),\qquad \text{and} \qquad
	(D_{ij}^{\alpha }(g))^{\dagger}=(D_{ij}^{\alpha }(g))^{-1},
	\ee
	where $i,j=1,2,\ldots,d_{\alpha }$.
\end{notation}

The matrix elements $D_{ij}^{\alpha }(g)$ will play a crucial role in the
following.

\begin{definition}
	\label{ad1}
	Let $D:G\rightarrow \Hom(V)$ be an unitary representation of a finite group $G
	$ such that the operators $D(g),$ $\ g\in G$ are linearly independent i.e. $%
	\dim A_{V}[G]=\left\vert G\right\vert $ and let $D^{\alpha }:G\rightarrow
	\Hom(V^{\alpha })$ be all inequivalent, irreducible representations of $G$
	described in Notation~\ref{not1} above. Define
	\be
	E_{ij}^{\alpha }=\frac{d_{\alpha }}{|G|}\sum_{g\in G}D_{ji}^{\alpha
	}(g^{-1})D(g),
	\ee
	where $\alpha =1,2,\ldots,r,\quad i,j=1,2,\ldots,d_{\alpha },\quad
	E_{ij}^{\alpha }\in A_{V}[G]\subset \Hom(V)$.
\end{definition}

The operators have the following properties:

\begin{theorem}
	\label{at1}
	\begin{enumerate}
	\item There are exactly $\left\vert G\right\vert$ nonzero operators $
	E_{ij}^{\alpha }$ and
	\be
	D(g)=\sum_{ij\alpha }D_{ij}^{\alpha }(g)E_{ij}^{\alpha }.
	\ee
	
	\item The operators $E_{ij}^{\alpha }$ are orthogonal with respect to the
	Hilbert-Schmidt scalar product in the space $\Hom(V)$.
	\be
	(E_{ij}^{\alpha },E_{kl}^{\beta })=\tr\left[ (E_{ij}^{\alpha }) ^{\dagger}E_{kl}^{\beta
	}\right] =m_{\alpha }\delta ^{\alpha \beta }\delta _{ik}\delta _{jl},\qquad
	m_{\alpha }\geq 1,
	\ee
	where $m_{\alpha }$ \ is equal to the multiplicity of the irreducible
	representation $D^{\alpha }$ in $D$ and it does not depend on $i,j=1,2,\ldots,d_{\alpha }$.
	
	\item The operators $E_{ij}^{\alpha }$ satisfy the following composition rule%
	\be
	E_{ij}^{\alpha }E_{kl}^{\beta }=\delta ^{\alpha \beta }\delta
	_{jk}E_{il}^{\alpha },
	\ee
	in particular $E_{ii}^{\alpha }$ are orthogonal projections.
	\end{enumerate}
\end{theorem}

\begin{remark}
	\label{ar1}
	From item 2 of above theorem it follows that the expressions
	\be
	E_{ij}^{\alpha }=\frac{d_{\alpha}}{|G|}\sum_{g\in G}D_{ji}^{\alpha }(g^{-1})D(g)
	\ee
	describe the transformation of orthogonalization of operators $D(g),$ $\ g\in G$
	in the space $\Hom(V)$ with the Hilbert-Schmidt scalar product.
\end{remark}

The operators  $E_{ii}^{\alpha }$ are not only orthogonal
projections onto their proper subspaces in $V$ but they are also orthogonal
with respect to the Hilbert-Schmidt scalar product in the space $\Hom(V)$.
The basis $\{E_{ij}^{\alpha }\}$ plays essential role when $D:G\rightarrow \mathbb{C}[ G]$ is the regular representation. In this case the properties of the
basis $\{E_{ij}^{\alpha }\}$ expresses the well-known fact that the group
algebra $\mathbb{C} [G]$ is a direct sum of simple matrix algebras generated by the
irreducible representations of the group $G$. It is always possible to
construct the operators $E_{ij}^{\alpha }$ even if the operators $D(g)$ are
not linearly independent but in this case some of them will be zero.

\section{Proof of the SDP-related lemmas}\label{SDPlemmas}
\subsection{Proof of Lemma 4:}\label{SDPlemmas.1}
\begin{proof}
	The symmetries in our problem suggest that we may take $\Theta_{\overline{a}}$ as an element of the algebra $\mathbb{C}\left[S(n-2)\right] $. Thus,  $\Theta_{\overline{a}}=\sum_{\alpha}x_{\alpha}P_{\alpha}$, where $P_{\alpha}$ are Young projectors and $x_{\alpha}\in \mathbb{R}_+$  which ensures that first constraint from~\eqref{cons} is automatically satisfied. Using this argumentation we can rewrite the second constraint from~\eqref{cons} restricted to an irrep labelled by $\alpha \vdash n-2$ as
	\be
	\sum_{a=1}^{n-1} P^+_{a,n}\ot \Theta_{\overline{a}}(\alpha)=x_{\alpha}\sum_{a=1}^{n-1}V(a,n-1)P^+_{n-1,n}\ot P_{\alpha}V(a,n-1)=\frac{x_{\alpha}}{d} \eta(\alpha),
	\ee
	where $\eta(\alpha)$ are introduced in Theorem~\ref{m1}. Eigenvalues of the operator $\frac{1}{d}\eta(\alpha)$ are equal to $\frac{1}{d}\gamma_{\mu}(\alpha)$, where numbers $\gamma_{\mu}(\alpha)$ are eigenvalues of $\eta(\alpha)$ given in Theorem~\ref{m2}.
	To ensure that $\forall \alpha  \ \frac{x_{\alpha}}{d}\eta(\alpha)\leq \mathbf{1}_{\alpha}$ we take
	\be
	\label{max}
	x_{\alpha}=\mathop{\min}\limits_{ \mu \in \alpha} \frac{1}{\frac{1}{d}\gamma_{\mu}(\alpha)}=d\mathop{\min}\limits_{ \mu \in \alpha} \frac{1}{\gamma_{\mu}(\alpha)}.
	\ee
	To obtain the minimum it suffices to insert $\gamma_{\mu^{\ast}}(\alpha)$, which is the maximal possible eigenvalue of the operator $\eta(\alpha)$ for some particular Young frame $\mu \vdash n-1 $ obtained from $\alpha \vdash n-2$ by adding one box in the proper way.
	Inserting the optimal form of operators $\Theta_{\overline{a}}$ into equation~\eqref{opt}, we get
	\be
	p^\star=\frac{1}{d^{N+1}}\sum_{a=1}^N \tr \left(\sum_{\alpha}x_{\alpha}P_{\alpha} \right)=\frac{N}{d^{N}}\sum_{\alpha} \frac{1}{\gamma_{\mu^{\ast}}(\alpha)}\tr P_{\alpha}=\frac{N}{d^{N}}\sum_{\alpha} \frac{m_{\alpha}d_{\alpha}}{\gamma_{\mu^{\ast}}(\alpha)} = {\frac{1}{d^N}\sum_{\alpha}m_{\alpha}d_\alpha\min_{\mu \in \alpha}\frac{1}{\gamma_{\mu}(\alpha)}}.
	\ee

\end{proof}
\subsection{Proof of Fact 5:}\label{SDPlemmas.2}
\begin{proof}
	The operators $F_{\mu}(\alpha)$ are invariant under the action of $S(n-1)$ and thus under $S(n-2)$. The operator $V^{t_n}(n-1,n)$ is invariant under the action of $S(n-2)$. It follows that the composition $V^{t_n}(n-1,n)F_{\mu}(\alpha)$ is invariant under the action of $S(n-2)$. Moreover, using Fact~\ref{F3app} from Appendix~\ref{AppC} we have $\tr_{n-1,n}\left[V^{t_n}(n-1,n)F_{\mu}(\alpha) \right]\in \mathbb{C}[S(n-2)]$, and since it is invariant under the action of $S(n-2)$, it must be of the form $\bigoplus_{\beta \vdash n-2}y(\beta)P_{\beta}$, where $y_{\beta}\in \mathbb{C}$. However, by Theorem~\ref{m1} and Fact~\ref{BF2} from Appendix~\ref{AppB} we get 
	\be
	\begin{split}
 P_{\beta}V^{t_n}(n-1,n)F_{\mu}(\alpha)&=P_{\beta}V^{t_n}(n-1,n)M_{\alpha}P_{\mu}=P_{\beta}V^{t_n}(n-1,n)P_{\alpha}P_{\mu}\\
 &=\delta_{\alpha \beta} P_{\alpha}V^{t_n}(n-1,n)F_{\mu}(\alpha).
\end{split}
 \ee
 This implies that $\tr_{n-1,n}\left[V^{t_n}(n-1,n)F_{\mu}(\alpha) \right]=y_{\mu}(\alpha)P_{\alpha}$, and thus
 \be
 y_{\mu}(\alpha)=\frac{\tr\left[V^{t_n}(n-1,n)F_{\mu}(\alpha) \right] }{d_{\alpha}m_{\alpha}}=\frac{\tr\left[ V^{t_n}(n-1,n)M_{\alpha}P_{\mu}\right] }{d_{\alpha}m_{\alpha}}.
 \ee
 To get the final result, we apply Fact~\ref{BF2} once more
 \be
  y_{\mu}(\alpha)=\frac{\tr\left[V^{t_n}(n-1,n)P_{\alpha}P_{\mu} \right] }{d_{\alpha}m_{\alpha}}=\frac{\tr\left[P_{\mu}\left(P_{\alpha}\ot \mathbf{1} \right)  \right] }{d_{\alpha}m_{\alpha}}=\frac{m_{\mu}}{m_{\alpha}}.
 \ee
\end{proof}
\subsection{Proof of Lemma 6:}\label{SDPlemmas.3}

\begin{proof}
	Let us assume that the coefficients $x_{\mu^*}(\alpha)$ given in definition of the operator $\Omega$ in equation~\eqref{om12} are of the form $x_{\mu^*}(\alpha)=d\frac{m_{\alpha}}{m_{\mu^*}}$. One can easily see that $\Omega \geq 0$, and using Fact~\ref{Fop} we get
	\be
	\begin{split}
	\tr_{n-1,n}\left[P^+_{n-1,n}\Omega\right]&=\frac{1}{d}\tr\left[V^{t_n}(n-1,n)\Omega \right]=\sum_{\alpha}\frac{m_{\alpha}}{m_{\mu^*}}\tr\left[V^{t_n}(n-1,n)F_{\mu^*}(\alpha) \right]\\
	&=\sum_{\alpha}P_{\alpha}=\mathbf{1}_{1,\ldots,n-2},
	\end{split}
	\ee
	where $\mathbf{1}_{1,\ldots,n-2}$ denotes the identity operator defined on first $n-2$ subsystems. We thus see that the second constraint from expression~\eqref{dualA} is also fulfilled. Finally we can calculate quantity $p_{\star}$ given in equation~\eqref{low_star}:
	\be
	p_{\star}=\frac{1}{d^{N+1}}\tr \Omega=\frac{1}{d^N}\sum_{\alpha}\frac{m_{\alpha}}{m_{\mu^*}}\tr F_{\mu^*}(\alpha)=\frac{1}{d^N}\sum_{\alpha}m_{\alpha}^2\frac{d_{\mu^*}}{m_{\mu^*}}={\frac{1}{d^N}\sum_{\alpha}m_{\alpha}d_\alpha \min_{\mu \in \alpha}\frac{1}{\gamma_{\mu}(\alpha)}},
	\ee
	since $\tr F_{\mu^*}(\alpha)=d_{\mu^*}m_{\alpha}$ by Theorem~\ref{m1}.
\end{proof}
\subsection{Proof of Lemma 9}\label{SDPlemmas.4}
\begin{proof}
The symmetry of the problem suggests that optimal POVMs should be elements of the algebra $\mathbb{C}[S(n-2)]$. We represent them in the following form:
\be
\label{povmres}
\forall \ a=1,\ldots,N\quad \Theta_{\overline{a}}=\sum_{\alpha}u(\alpha)P_{\alpha,\overline{a}},\quad \text{with} \quad u(\alpha)=\frac{d^{N+1}g(N)m_{\alpha}}{Nd_{\alpha}},
\ee
where $g(N)=1/\sum_{\nu}m_{\nu}^2$ for all  $\nu \vdash n-1$, and the above sum runs over all allowed irreps of $S(n-2)$. By $P_{\alpha, \overline{a}}$ we denote Young projectors onto irreps of $S(n-2)$, but defined on every subsystem except $n$-th and $a$-th. Since all coefficients $u(\alpha)\geq 0$, then the first constraint from~\eqref{cons1} is satisfied. We choose the optimal state defined through $X_A$ from~\eqref{cons1} is of the form:
\be
\label{optstate}
X_A=\sum_{\mu}c_{\mu}P_{\mu},\quad \text{where} \quad c_{\mu}=\frac{d^{N}g(N)m_{\mu}}{d_{\mu}},
\ee
where sum runs over all allowed irreps of $S(n-1)$. We see that
\be
\tr X_A=\sum_{\mu}c_{\mu}\tr P_{\mu}=d^N\frac{\sum_{\mu}m_{\mu}^2}{\sum_{\nu}m_{\nu}^2}=d^N,
\ee
so the constraint on the trace of~\eqref{cons1} is fulfilled. Moreover, we have that
\be
\forall \mu\in \alpha \quad u(\alpha)=\frac{d}{\gamma_{\mu}(\alpha)}c_{\mu},
\ee
where numbers $\gamma_{\mu}(\alpha)$ are eigenvalues given by Proposition~\ref{m2}. Thus, the second constraint from~\eqref{cons1} is satisfied with the equality. For the above choices we compute the probability of success given in the statement of this Lemma plugging the choice of POVMs given in~\eqref{povmres} into~\eqref{pres}:
\be
\begin{split}
p^{\star}&=\frac{1}{d^{N+1}}\sum_{a=1}^N\tr\Theta_a=\frac{N}{d^{N+1}}\sum_{\alpha}u(\alpha)\tr P_{\alpha}=\frac{\sum_{\alpha}m_{\alpha}^2}{\sum_{\nu}m_{\nu}^2}=\frac{N}{N+d^2-1}=1-\frac{d^2-1}{N+d^2-1},
\end{split}
\ee
where we used the fact that $\tr P_{\alpha}=m_{\alpha}d_{\alpha}$, and plugged the ratio from Proposition~\ref{P23} in Appendix~\ref{app:multSn}.
\end{proof}

\subsection{Proof of Lemma 10}\label{SDPlemmas.5}

\begin{proof}
Since the operators $F_{\mu}(\alpha)$ are invariant under the action of $S(n-1)$ and  $\tr_{n} F_{\mu}(\alpha) \in \mathbb{C}[S(n-1)]$, the partial trace decomposes as $\tr_{n} F_{\mu}(\alpha)=\bigoplus_{\mu \vdash n-1} a_{\nu}P_{\nu}$, where $a_{\nu}\in \mathbb{C}$. On the other hand, by Theorem~\ref{m1} we have
\be
P_{\nu}F_{\mu}(\alpha)=P_{\nu}M_{\alpha}P_{\mu}=\delta_{\mu \nu}F_{\mu}(\alpha),
\ee
hence 
\be
\tr_{n} F_{\mu}(\alpha)=a_{\mu}P_{\mu},
\ee
and 
\be
a_{\mu}=\frac{\tr F_{\mu}(\alpha)}{\tr P_{\mu}}=\frac{m_{\alpha}d_{\mu}}{m_{\mu}d_{\mu}}=\frac{m_{\alpha}}{m_{\mu}},
\ee
where we use that $\tr F_{\mu}(\alpha)=m_{\alpha}d_{\mu}$ from Theorem~\ref{m1}.
\end{proof}
\subsection{Proof of Lemma 11}\label{SDPlemmas.6}

\begin{proof}
We represent the coefficients $x_{\mu}(\alpha)$ of $\Omega$ from ~\eqref{om1} as follows:
\be
x_{\mu}(\alpha)=\frac{d}{d^2+N-1}\gamma_{\mu}(\alpha),
\ee
where numbers $\gamma_{\mu}(\alpha)$ are eigenvalues of the PBT operator $\eta$ given in Proposition~\ref{m2}. Making use of Theorem~\ref{m1} we get
\be
\Omega=\sum_{\alpha}\sum_{\mu\in \alpha}x_{\mu}(\alpha)F_{\mu}(\alpha)=\frac{d}{N+d^2-1}\eta \geq 0
\ee
for any value $d$ and $N$, so the first constraint from~\eqref{cons2} is fulfilled, since $\eta \geq 0$. For the second constraint, making use of the symmetry of the problem it suffices to estimate it for $a=n-1$:
\be
\tr_{n-1,n}\left[P^+_{n-1,n}\Omega \right]=\frac{1}{N+d^2-1}\tr_{n-1,n}\left[V^{t_n}(n-1,n)\eta \right]=\mathbf{1}_{1\ldots n-2}, 
\ee
where $\mathbf{1}_{1\ldots n-2}$ denotes the identity operator acting on $n-2$ subsystems. Finally we need $\tr_n \Omega$:
\be
\label{a1}
\tr_n \Omega=\frac{d}{N+d^2-1}\sum_{\alpha}\sum_{\mu \in \alpha}\gamma_{\mu}(\alpha)\tr_n F_{\mu}(\alpha).
\ee
Using Lemma~\ref{Fop2} we reduce equation~\eqref{a1} to
\be
\label{trn}
\tr_n \Omega=\frac{d}{N+d^2-1}\sum_{\alpha}\sum_{\mu \in \alpha}\gamma_{\mu}(\alpha)\frac{m_{\alpha}}{m_{\mu}}P_{\mu}=\frac{d}{N+d^2-1}\sum_{\mu}\frac{1}{m_{\mu}}\sum_{\alpha \in \mu}\gamma_{\mu}(\alpha)m_{\alpha}P_{\mu},
\ee
where by $\alpha \in \mu$ we denote Young diagrams $\alpha$ of $n-2$ which can be obtained from Young diagrams $\mu$ of $n-1$ by removing one box in a proper way.  Now using the explicit form of $\gamma_{\mu}(\alpha)$ given in Proposition~\ref{m2} and Lemma~\ref{L26} from Appendix~\ref{app:multSn} we can simplify Eqn.~\eqref{trn}
\be
\tr_n \Omega=\frac{dN}{N+d^2-1}\sum_{\mu}\frac{1}{d_{\mu}}\sum_{\alpha \in \mu}d_{\alpha}P_{\mu}=\frac{dN}{N+d^2-1}\sum_{\mu}P_{\mu}=\frac{dN}{N+d^2-1}\mathbf{1}_{1\ldots n-1},
\ee
where $\mathbf{1}_{1\ldots n-1}$ is identity operator defined on $n-1$ subsystems. Now taking $b=\frac{1}{d^N}\frac{N}{N+d^2-1}$, we satisfy the third constraint from~\eqref{cons2}
\be
p_{\star}=d^Nb=1-\frac{d^2-1}{N+d^2-1}=p^{\star}.
\ee
\end{proof}

\end{document}